\documentclass[letterpaper]{article}

\usepackage{arxiv}
\usepackage{newtxtext} 
\usepackage[utf8]{inputenc} % allow utf-8 input
\usepackage[T1]{fontenc}    % use 8-bit T1 fonts
\usepackage{hyperref}       % hyperlinks
\usepackage{url}            % simple URL typesetting
\usepackage{booktabs}       % professional-quality tables
\usepackage{amsfonts}       % blackboard math symbols
\usepackage{nicefrac}       % compact symbols for 1/2, etc.
\usepackage{microtype}      % microtypography
\usepackage{lipsum}		% Can be removed after putting your text content
\usepackage{graphicx}
\usepackage{natbib}
\usepackage{doi}
\usepackage{todonotes}
\usepackage{comment}
\usepackage{subcaption}
\usepackage{glossaries}
\usepackage[switch]{lineno} 
\usepackage[titletoc]{appendix}

\allowdisplaybreaks

\usepackage{multirow}
\usepackage{float}
\usepackage{amsmath, latexsym, epsf, amssymb, framed, bbm, bm, amsthm, upref, bbold, algorithm, algorithmic}
\usepackage{xcolor}
\usepackage[capitalize]{cleveref}
\usepackage{enumitem}

\usepackage{makecell}

\allowdisplaybreaks

\DeclareMathOperator*{\argmax}{arg\,max}
\DeclareMathOperator*{\argmin}{arg\,min}

\fancypagestyle{researchstyle}{
    \fancyhf{}  % Clear default header/footer
    \fancyhead[R]{ }  % Right-aligned name
    \fancyhead[L]{}  % Left-aligned section title 
    \fancyfoot[C]{\thepage}
}
\title{Estimate then Predict: Convex Formulation for Travel Demand Forecasting}

\author{
    Youngseo Kim\textsuperscript{a, *}, 
    Gioele Zardini\textsuperscript{b}, 
    Samitha Samaranayake\textsuperscript{c}, 
    Soroosh Shafiee\textsuperscript{d} \\
    \\
    \textsuperscript{a}Civil and Environmental Engineering, UCLA, Los Angeles, CA, USA\\
    \texttt{youngseo@ucla.edu} \\
    \textsuperscript{b}Laboratory for Information and Decision Systems, MIT, Cambridge, MA, USA\\
    \texttt{gzardini@mit.edu} \\
    \textsuperscript{c}Civil and Environmental Engineering, Cornell University, Ithaca, NY, USA\\
    \texttt{samitha@cornell.edu} \\
    \textsuperscript{d}Operations Research and Information Engineering, Cornell University, Ithaca, NY, USA\\
    \texttt{shafiee@cornell.edu} \\ 
    \textsuperscript{*}Corresponding author. Email: \texttt{youngseo@ucla.edu}
}

\renewcommand{\headeright}{ }
\renewcommand{\undertitle}{ }
\renewcommand{\shorttitle}{ }

\newtheorem{theorem}{Theorem}
\newtheorem{assumption}{Assumption}
\newtheorem{corollary}{Corollary}

\newtheorem{lemma}{Lemma}

\newacronym{a:riarum}{RI-ARUM}{Rational Inattention - Addictive Random Utility Model}
\newacronym{a:asc}{ASC}{Alternative Specific Constant}
\newacronym{a:bpr}{BPR}{Bureau of Public Roads}
\newacronym{a:socp}{SOCP}{Second Order Cone Programming}
\newacronym{a:soc}{SOC}{Second Order Cone}
\newacronym{a:lp}{LP}{Linear Programming}
\newacronym{a:ml}{ML}{Machine Learning}
\newacronym{a:kl}{KL}{Kullback–Leibler}
\newacronym{a:nl}{NL}{Nested Logit}
\newacronym{a:mnl}{MNL}{Multinomial Logit}
\newacronym{a:psl}{PSL}{Path-size Logit}
\newacronym{a:kkt}{KKT}{Karush–Kuhn–Tucker}
\newacronym{a:od}{OD}{Origin-Destination}

\newcommand{\Pset}[0]{\ensuremath{\mathcal{P}}}

\newcommand{\betavec}[0]{\ensuremath{\boldsymbol{\beta}}}
\newcommand{\alphavec}[0]{\ensuremath{\boldsymbol{\alpha}}}
\newcommand{\gammavec}[0]{\ensuremath{\boldsymbol{\gamma}}}

\newcommand{\nest}[0]{\ensuremath{\mathcal{N}}}
\newcommand{\nestset}[0]{\ensuremath{\Pi (\mathcal{M})}}
\newcommand{\nestofm}[0]{\ensuremath{\mathcal{N}}}

\newcommand{\nestsetfull}[0]{\ensuremath{\Pi (\mathcal{M}_{ij})}}

\newcommand{\jprobr}[0]{\ensuremath{p_{ijmr}}}
\newcommand{\jprobnest}[0]{\ensuremath{p_{ij\nest}}}
\newcommand{\jprobm}[0]{\ensuremath{p_{ijm}}}
\newcommand{\jprobj}[0]{\ensuremath{p_{ij}}}

\newcommand{\iset}[0]{\ensuremath{\mathcal{I}}}
\newcommand{\jset}[0]{\ensuremath{\mathcal{J}}}
\newcommand{\mset}[0]{\ensuremath{\mathcal{M}}}

\newcommand{\kset}[0]{\ensuremath{\mathcal{K}}}
\newcommand{\qset}[0]{\ensuremath{\mathcal{Q}}}

\newcommand{\msetfull}[0]{\ensuremath{\mathcal{M}_{ij}}}
\newcommand{\rset}[0]{\ensuremath{\mathcal{R}}}
\newcommand{\rsetfull}[0]{\ensuremath{\mathcal{R}_{ijm}}}
\newcommand{\aset}[0]{\ensuremath{\mathcal{A}}}
\newcommand{\asetfull}[0]{\ensuremath{\mathcal{A}_{ijmr}}}

\newcommand{\jprobvecr}[0]{\ensuremath{\mathbf{p}_{\substack{\phantom{=}\\[-0.3ex]\iset \jset \mset \rset}}}}
\newcommand{\jprobvecnest}[0]{\ensuremath{\mathbf{p}_{\substack{\phantom{=}\\[-0.3ex] \iset \jset \nestset}}}}
\newcommand{\jprobvecgivennest}[0]{\ensuremath{\mathbf{p}_{\substack{\phantom{=}\\[-0.3ex] \iset \jset \nest}}}}
\newcommand{\jprobvecm}[0]{\ensuremath{\mathbf{p}_{\substack{\phantom{=}\\[-0.3ex] \iset \jset \mset }}}}
\newcommand{\jprobvecj}[0]{\ensuremath{\mathbf{p}_{\substack{\phantom{=}\\[-0.3ex] \iset \jset }}}}

\newcommand{\obsjprobvecm}[0]{\ensuremath{\widehat{\mathbf{p}}_{\substack{\phantom{=}\\[-0.3ex] \iset \jset \mset }}}}
\newcommand{\obsjprobvecj}[0]{\ensuremath{\widehat{\mathbf{p}}_{\substack{\phantom{=}\\[-0.3ex] \iset \jset }}}}

\newcommand{\cprobvecr}[0]{\ensuremath{\mathbf{p}_{\rset|\iset \jset \mset}}}
\newcommand{\cprobvecnest}[0]{\ensuremath{\mathbf{p}_{\nestset |\iset \jset}}}
\newcommand{\cprobvecmnest}[0]{\ensuremath{\mathbf{p}_{\mset| \iset \jset \nestset}}}

\newcommand{\cprobvecm}[0]{\ensuremath{\mathbf{p}_{\mset |\iset \jset}}}
\newcommand{\cprobvecj}[0]{\ensuremath{\mathbf{p}_{\jset|\iset}}}

\newcommand{\obsjprobnest}[0]{\ensuremath{\widehat{p}_{ij\nest}}}
\newcommand{\obsjprobm}[0]{\ensuremath{\widehat{p}_{ijm}}}
\newcommand{\obsjprobj}[0]{\ensuremath{\widehat{p}_{ij}}}

\newcommand{\optjprobr}[0]{\ensuremath{p_{ijmr}}}
\newcommand{\optjprobnest}[0]{\ensuremath{p_{ij\nest}}}
\newcommand{\optjprobm}[0]{\ensuremath{p_{ijm}}}
\newcommand{\optjprobj}[0]{\ensuremath{p_{ij}}}

\newcommand{\optcprobr}[0]{\ensuremath{p_{r|ijm}}}

\newcommand{\optcprobj}[0]{\ensuremath{p_{j|i}}}
\newcommand{\optcprobnest}[0]{\ensuremath{p_{\nest|ij}}}
\newcommand{\optcprobmnest}[0]{\ensuremath{p_{m|ij\nest}}}

\newcommand{\obscprobj}[0]{\ensuremath{\widehat{p}_{j|i}}}
\newcommand{\obscprobnest}[0]{\ensuremath{\widehat{p}_{\nest|ij}}}
\newcommand{\obscprobmnest}[0]{\ensuremath{\widehat{p}_{m|ij\nest}}}

\newcommand{\obscprobvecm}[0]{\ensuremath{\widehat{\mathbf{p}}_{\mset |\iset \jset}}}
\newcommand{\obscprobvecj}[0]{\ensuremath{\widehat{\mathbf{p}}_{\jset|\iset}}}
\newcommand{\obscprobvecnest}[0]{\ensuremath{\widehat{\mathbf{p}}_{\nestset|\iset \jset}}}

\newcommand{\cprobr}[0]{\ensuremath{p_{r|ijm}}}
\newcommand{\cprobnest}[0]{\ensuremath{p_{\nest|ij}}}
\newcommand{\cprobmnest}[0]{\ensuremath{p_{m|ij\nest}}}
\newcommand{\cprobm}[0]{\ensuremath{p_{m|ij}}}
\newcommand{\cprobj}[0]{\ensuremath{p_{j|i}}}

\newcommand{\lag}[0]{\ensuremath{\mathcal{L}}}

\newcommand{\satis}[0]{\ensuremath{S}}
\newcommand{\utilvec}[0]{\ensuremath{\mathbb{V}}}

\newcommand{\scalepar}[0]{\ensuremath{\hat{\lambda}}}

\newcommand{\scalej}[0]{\ensuremath{\theta_{\text{dest}}}}
\newcommand{\scalem}[0]{\ensuremath{\theta_{\text{mode}}}}
\newcommand{\scaler}[0]{\ensuremath{\theta_{\text{route}}}}

\newcommand{\pathsizepar}[0]{\ensuremath{\widehat{\psi}_{ijmr}}}
\newcommand{\nestpar}[0]{\ensuremath{\hat{\tau}_{\nest}}}

\newcommand{\duali}[0]{\ensuremath{\lambda_i}}
\newcommand{\dualj}[0]{\ensuremath{\mu_{ij}}}
\newcommand{\dualm}[0]{\ensuremath{\nu_{ijm}}}
\newcommand{\dualM}[0]{\ensuremath{\kappa_{ij\nest}}}

\newcommand{\optduali}[0]{\ensuremath{\lambda_i}}
\newcommand{\optdualj}[0]{\ensuremath{\mu_{ij}}}
\newcommand{\optdualm}[0]{\ensuremath{\nu_{ijm}}}
\newcommand{\optdualM}[0]{\ensuremath{\kappa_{ij\nest}}}

\newcommand{\optscalej}[0]{\ensuremath{\widehat{\theta}_{dest}}}
\newcommand{\optscalem}[0]{\ensuremath{\widehat{\theta}_{mode}}}

\newcommand{\optnestpar}[0]{\ensuremath{\widehat{\tau}_\nest}}

\newcommand{\obsori}[0]{\ensuremath{\widehat{O}_i}}
\newcommand{\obsprobi}[0]{\ensuremath{\widehat{p}_i}}
\newcommand{\obsprobivec}[0]{\ensuremath{\widehat{\mathbf{p}}_{\iset}}}

\newcommand{\arcflow}[0]{\ensuremath{f^m_a}}
\newcommand{\ind}[0]{\ensuremath{\widehat{\delta}^{a}_{ijmr}}}

\newcommand{\obstotal}[0]{\ensuremath{\widehat{N}}}
\newcommand{\obstrip}[0]{\ensuremath{\widehat{T}_{ij}}}
\newcommand{\obstripm}[0]{\ensuremath{\widehat{T}_{ijm}}}

\newcommand{\cutilj}[0]{\ensuremath{V_{j|i}}}
\newcommand{\cutilm}[0]{\ensuremath{V_{m|ij}}}
\newcommand{\cutilr}[0]{\ensuremath{V_{r|ijm}}}

\newcommand{\errj}[0]{\ensuremath{\varepsilon_{j|i}}}
\newcommand{\errm}[0]{\ensuremath{\varepsilon_{m|ij}}}
\newcommand{\errr}[0]{\ensuremath{\varepsilon_{r|ijm}}}

\newcommand{\tripattrpar}[0]{\ensuremath{\beta_k}}
\newcommand{\tripmattrpar}[0]{\ensuremath{\beta_q}}

\newcommand{\opttripattrpar}[0]{\ensuremath{\widehat{\beta}_k}}
\newcommand{\opttripmattrpar}[0]{\ensuremath{\widehat{\beta}_q}}

\newcommand{\tripattr}[0]{\ensuremath{\widehat{X}^k_{ij}}}
\newcommand{\tripmattr}[0]{\ensuremath{\widehat{X}^q_{ijm}}}

\newcommand{\satistrip}[0]{\ensuremath{S_{ij}}}
\newcommand{\satistripm}[0]{\ensuremath{S_{ijm}}}

\newcommand{\intlatency}[0]{\ensuremath{\int^{\arcflow}_{0} \latency(w) dw}}

\newcommand{\latency}[0]{\ensuremath{g^m_a}}
\newcommand{\rcost}[0]{\ensuremath{g_{ijmr}}}

\newcommand{\beckmann}[0]{\ensuremath{ \sum_{i \in \iset, j \in \jset} \sum_{m \in \mset, r \in \rset} \sum_{a \in \asetfull} \intlatency}}

\newcommand{\auxmain}[0]{\ensuremath{s_a}}

\newcommand{\expcone}[0]{\ensuremath{{\cal{K}}_{exp}}}

\begin{document}
\maketitle
% \vspace{-3em}

% \linenumbers

\newcommand{\YS}[1]{{\color{black} {#1}}}
\newcommand{\SR}[1]{{\color{black} {#1}}} % second revision
\newcommand{\TR}[1]{{\color{black} {#1}}} % thrid revision
\newcommand{\FR}[1]{{\color{black} {#1}}} % fourth revision
\newcommand{\HL}[1]{{\color{black} {#1}}}

\newcommand{\GZ}[1]{{\color{green} {#1}}}
% \newpage

\begin{abstract}
Travel demand forecasting models play crucial roles in evaluating large-scale infrastructure projects, such as the construction of new roads or transit lines. 
While combined modeling approaches have been explored as a solution to overcome the problem of input/output discrepancies in a sequential four-step modeling process, previous attempts at combined models have encountered challenges in real-world applications, primarily due to their limited behavioral richness, computational tractability\HL{, and the lack of a unified framework for parameter estimation}. 
In this study, we propose a novel convex programming approach and present a key theorem \HL{demonstrating that the optimal primal solution satisfies a hierarchical extended logit model, while the optimal dual variables recover the taste coefficients, alternative-specific constants, and scale parameters.}
This model is specifically designed to capture correlations existing in travelers' choices, including similarities among transport modes and route overlaps. 
\HL{The convex structure provides global optimality guarantees and enables the use of efficient off-the-shelf conic solvers.}
The advantages of our proposed model are twofold. 
First, it provides a single unifying rationale (i.e., utility/entropy maximization) that is valid across all steps. 
\HL{Second, its combined nature enables the systematic incorporation of observed data by estimating parameters while accounting for network topology and flow dynamics, thereby providing a behaviorally richer and internally consistent representation of travel behavior.}
Our convex programming approach shows promise for improving the applicability and scalability of travel demand forecasting, thereby supporting infrastructure planning and decision-making.
\HL{Numerical experiments on eight benchmark transportation networks demonstrate the scalability and computational efficiency of the proposed formulation across networks of varying sizes.}
\end{abstract}

% keywords can be removed

\section{Introduction}

\YS{Travel demand modeling plays a critical role in assessing large-scale infrastructure projects, such as the construction of new roads or transit lines.}
To evaluate a new transportation system and estimate future demand, the travel demand modeling process is typically conducted in two stages~\citep{mcnally2007four}. 
The first stage is referred to as \emph{parameter estimation} (or data assimilation), during which the models used in each step are evaluated, calibrated, and validated using data from the existing transportation system\HL{, with model performance assessed by how well the resulting travel patterns reproduce observed reality}.
The second stage is \emph{future estimation}. 
At this stage, the \HL{projected} future productions and attractions (or just productions) are allocated to the new transportation network, which \HL{may} include newly constructed roads or transit lines. 
\YS{While this two-stage process of \textit{estimation then prediction} remains fundamental, the traditional four-step sequential modeling process can be replaced with alternative modeling approaches.}

\YS{The traditional sequential four-step model is the cornerstone of travel demand forecasting, employing distinct models for each stage: trip generation (often using regression), trip distribution (typically with gravity models), modal split (determined by logit models), and traffic assignment (based on Wardrop's user equilibrium)~\citep{de2011modelling}. 
This integrated process is the primary tool for forecasting future demand and assessing the performance of a transportation system, especially when evaluating large-scale infrastructure projects.
However, the model's sequential nature introduces a fundamental weakness: a lack of internal consistency. 
Because the steps are solved in isolation, the outputs of one stage may contradict the inputs of another. 
For instance, travel times assumed during trip distribution often fail to match the congested travel times produced by the final traffic assignment step~\citep{boyce2002sequential, garrett1996transportation}. 
Furthermore, this structure allows errors from early steps to propagate and compound, potentially leading to significant inaccuracies in the final forecast.
While feedback loops are often used to address these inconsistencies, by feeding equilibrated travel times back into earlier steps, this approach lacks a convergence guarantee and often relies on time-consuming, arbitrary adjustments by experts. 
} 

\YS{Recognizing these limitations, researchers have developed combined models that simultaneously consider traveler choices across multiple stages. 
As illustrated in \cref{fig:comparison}, a sequential model equilibrates only the route choice (a), whereas our combined model can equilibrate route, mode and destination choices together (b). This integrated approach provides more coherent outcomes and better reflects how real-world congestion influences the full spectrum of travel decisions\HL{ (e.g., congestion affects not only route choice but also mode and destination choices)}.
\HL{Although theoretical advances in combined convex programming were established as early as the 1970s, the practical adoption of combined models has been slow because of the lack of standard software and institutional inertia~\citep{mcnally2007four, nie2021full}. This gap reflects two deeper, unresolved challenges in the field: (i) the persistent trade-off between behavioral richness and computational tractability in real-world networks and (ii) the lack of frameworks for joint parameter estimation.}
To date, models with desirable convex properties have often achieved tractability at the cost of simplistic behavioral assumptions \HL{(e.g., reliance on the basic multinomial logit model)} or have failed to fully exploit convexity in the solution method~\citep{oppenheim1995urban, yao2014general, zhou2009}. 
This paper bridges that gap by introducing a combined model that achieves a higher degree of behavioral richness while retaining a convex programming formulation, thus providing a practical and powerful alternative to existing approaches.
}

\FR{
Most importantly, this paper bridges a long-standing gap in travel demand modeling: the joint estimation of behavioral parameters and the characterization of network flows.
Prior work has established convex formulations for travel demand modeling when behavioral parameters are fixed. 
However, jointly estimating behavioral parameters and equilibrium flows generally introduces nonlinear interactions that can destroy convexity.
We address this challenge by leveraging duality theory to maintain convexity.
\HL{The key insight is to recover estimates of behavioral parameters (e.g., taste coefficients, alternative-specific constants, and scale parameters) as dual variables, thereby avoiding the bilinear terms that would otherwise make the formulation non-convex.}
Recent advances in related fields, particularly in information theory and convex optimization, have introduced key developments, including the duality between maximum likelihood estimation and logit-family models \citep{muller2022discrete}, as well as conic reformulation techniques and modern interior-point solvers \citep{aps2024mosek}.
Building on recent advances in convex optimization and information theory, this paper makes a novel contribution to the transportation field by integrating them into a unified framework for travel demand modeling.
}

\begin{figure}[!tb]
    \begin{subfigure}[b]{\textwidth}
        \centering \includegraphics[width=0.8\textwidth]{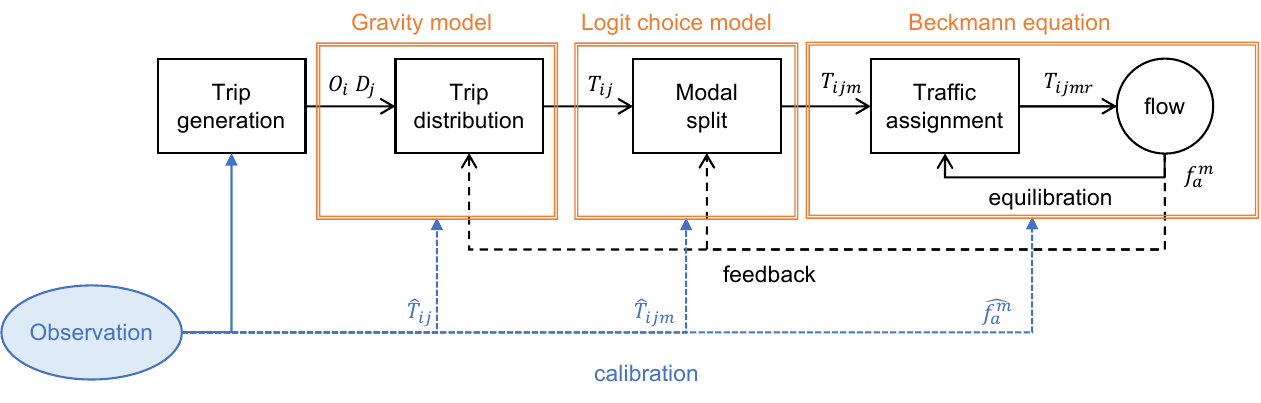}
        \caption{Traditional four-step travel demand modeling}
        \label{fig:sub1}
    \end{subfigure}
    \begin{subfigure}[b]{\textwidth}
        \centering \includegraphics[width=0.8\textwidth]{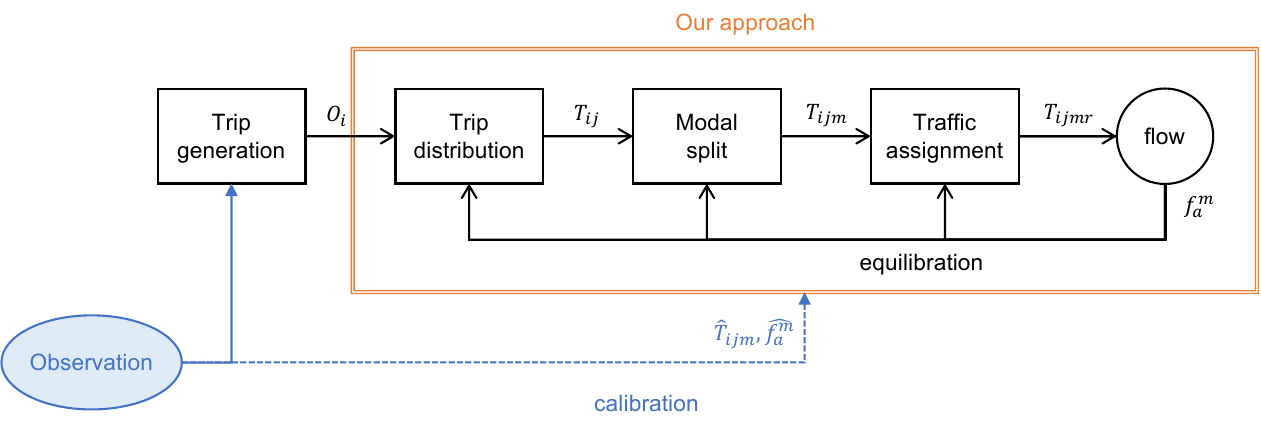}
        \caption{Combined travel demand modeling}
        \label{fig:sub2}
    \end{subfigure}
    \caption{Traditional four-step modeling (adapted from~\citet{mcnally2007four}) \textit{vs} our combined travel demand modeling.}
    \label{fig:comparison}
\end{figure}

\subsection{Literature review}

\FR{This section reviews the literature on convex formulations and parameter estimation for travel demand modeling, highlighting the gap between tractable equilibrium modeling and behavioral parameter estimation.}

\subsubsection{Convex reformulations of the four-step travel demand model}

\YS{The traditional four-step model is defined by two key inputs: the transportation system (i.e., network graphs with attributes like capacity and speed) and the activity system (e.g., socioeconomic and land use data). This framework divides the complex problem of travel demand into a sequence of distinct steps. Namely, it unfolds across four sequential steps: trip generation, which estimates trip productions ($O_i$) and attractions ($D_j$); trip distribution, which forms an Origin-Destination (OD) matrix ($T_{ij}$); modal split, which determines mode-specific OD matrices ($T_{ijm}$); and finally, traffic assignment, which assigns modal trips to routes to find a user equilibrium, often using the Frank-Wolfe algorithm \citep{mcnally2007four, fukushima1984modified}.
The model's sequential nature is its fundamental weakness. It can be viewed as an alternating minimization scheme where each stage optimizes its objective based on \emph{fixed} inputs from prior stages. For instance, trip distribution and modal split rely on fixed travel times that often fail to match the congested travel times produced by the final traffic assignment. This reliance on fixed variables compromises optimality and prevents the model from capturing the true interdependencies within the system.

Combined convex models resolve this issue by jointly optimizing all variables simultaneously, thus avoiding this loss of optimality and providing a more accurate representation of network equilibrium. The concepts of entropy and the Beckmann equation are fundamental to creating a combined convex model of travel demand. 
Entropy, introduced in information theory by \citet{shannon1948mathematical}, plays a central role in choice modeling, where it is interpreted either as the value of choice variety \citep{miyagi1996direct} or as the cost of information acquisition in rational inattention \citep{fosgerau2020discrete}. 
For a comprehensive overview of entropy applications, see \citet{fang1997entropy}.
Complementing this, the Beckmann equation provides a powerful convex optimization formulation for the traffic assignment problem, elegantly ensuring that link travel costs correspond to Wardrop's user equilibrium \citep{beckmann1956studies}. 
Building upon the influential studies of \citet{wilson1969use}, \citet{anas1983discrete}, and \citet{beckmann1956studies}, the key innovation in the literature, summarized in Table~\ref{tab:literature}, has been to integrate both entropy functions and the Beckmann equation into a single objective function. 
}

The specific objectives are selected based on the steps that researchers aim to combine. The existing literature can be categorized into four groups: (a) combining trip distribution and traffic assignment \SR{\citep{evans1976derivation, florian1975combined, oppenheim1993equilibrium, yang2001simultaneous}}, (b) combining modal split and traffic assignment \citep{florian1977traffic, abdulaal1979methods, fernandez1994network, garcia2005network}, (c) combining trip distribution, modal split, and traffic assignment \citep{florian1978, friesz1981equivalent}, and (d) combining trip generation, distribution, modal split, and traffic assignment \citep{safwat1978, oppenheim1995urban, zhou2009}.

\begin{table}[ht!]
\caption{Literature on Combined Models for Four-Step Travel Demand Modeling}
\centering
\label{tab:literature}
\begin{tabular}{@{}lcccc@{}}
\toprule
\textbf{}                                        & \textbf{\begin{tabular}[t]{@{}c@{}}Trip \\ generation $^*$ \\ (Regression/ \\ Gravity) \\ \end{tabular}} & \textbf{\begin{tabular}[t]{@{}c@{}}Trip  \\ distribution\\ (Gravity \\ model) \end{tabular}} & \textbf{\begin{tabular}[t]{@{}c@{}}Modal \\ split\\ (Logit \\ model)\end{tabular}} & \textbf{\begin{tabular}[t]{@{}c@{}}Traffic \\ assignment\\ (User \\ equilibrium)\end{tabular}} \\ \midrule
\textbf{Framework}                 & \textbf{ }                & \textbf{Entropy}                                                                     & \textbf{Entropy}                                                                       & \textbf{\SR{Entropy}/Beckmann}                                                                                   \\ \midrule
\citet{wilson1969use}            &                          & O                                                                                    &                                                                                             &                                                                                                     \\
\citet{anas1983discrete}         &                          &                                                                                      & O                                                                                           &                                                                                                     \\
\citet{beckmann1956studies}      &                          &                                                                                      &                                                                                             & O                                                                                                   \\
\midrule
\citet{evans1976derivation}      &                          & O                                                                                    &                                                                                             & O                                                                                                   \\
\citet{florian1975combined}      &                          & O                                                                                    &                                                                                             & O                                                                                                   \\
\citet{oppenheim1993equilibrium} &                          & O                                                                                    &                                                                                             & O                                                                                                   \\\SR{\citet{yang2001simultaneous}}     &                          &  O                                            &                                                                                            & O                                                                                                \\
\citet{florian1977traffic}       &                          &                                                                                      & O                                                                                           & O                                                                                                   \\
\citet{abdulaal1979methods}      &                          &                                                                                      & O                                                                                           & O                                                                                                   \\
\citet{fernandez1994network}     &                          &                                                                                      & O                                                                                           & O                                                                                                   \\
\citet{garcia2005network}        &                          &                                                                                      & O                                                                                           & O                                                                                                   \\
\citet{florian1978}              &                          & O                                                                                    & O                                                                                           & O                                                                                                   \\
\citet{friesz1981equivalent}     &                          & O                                                                                    & O                                                                                           & O                                                                                                   \\
\citet{safwat1978}              & O                        & O                                                                                    & O                                                                                           & O                                                                                                   \\ 
\citet{oppenheim1995urban}       & O                        & O                                                                                    & O                                                                                           & O                                                                                                   \\
\citet{yao2014general}           & O                        & O                                                                                    &                                                                                             & O                                                                                                   \\
\citet{zhou2009}                 & O                        & O                                                                                    & O                                                                                           & O                                                                                                   \\
\textbf{Ours}                                    & \textbf{O}               & \textbf{O}                                                                           & \textbf{O}                                                                                  & \textbf{O}                                                                                          \\ \bottomrule
\end{tabular} \\
\flushleft{\footnotesize{$^*$ Trip generation can be integrated in two ways. The first approach uses a regression model to estimate only the origin demand and then determines the origin-destination demand pair in conjunction with the trip distribution process. The second method incorporates a regression model to generate origin and destination demands separately. The work of \cite{safwat1978} exemplifies the second approach, whereas other methods, including ours, typically fall under the first category.}}
\end{table}

A handful of research has focused on framing all steps of the travel decision process – including destination, mode, route choices, and their simultaneous interactions – within the \textit{utility maximization} paradigm \citep{oppenheim1995urban, zhou2009, yao2014general}. 
This approach views an individual traveler as a rational consumer of urban trips, with choices governed by random utility theory\HL{, assuming that travelers select the alternative that maximizes utility composed of deterministic and random components}. 
The resulting solution maximizes the collective utility of travelers, thereby effectively capturing aggregated travel patterns.
\YS{Note that} the terms entropy maximization and utility maximization can be used interchangeably, except during the trip generation step. 
Under the entropy maximization principle, the demand at both origins and destinations must be estimated during the trip generation step so that they can be provided to the trip distribution step to determine the full origin-destination demand pair. 
In contrast, under the utility maximization principle, only the origin demand is estimated during the trip generation step. The origin-destination demand pair is then determined, and the destination demand can be aggregated.
This method is particularly suitable given that estimates of trip production at origins are generally more reliable than estimates of trip attraction at destinations \citep{mcnally2007four}.
\YS{Our approach builds on utility maximization theory and enhances behavioral realism relative to existing combined models. Although nested mode choice and route-overlap correction have each been studied previously, prior formulations have generally considered these features separately or only in part. Specifically, \citet{safwat1978} and \citet{oppenheim1995urban} employ relatively simple mode-choice structures, \citet{zhou2009} does not account for route overlap, and \citet{yao2014general} focuses on a single mode (i.e., private car). Our hierarchical extended logit model brings these behavioral features together within a unified destination--mode--route framework.}

\HL{
\subsubsection{Parameter estimation via duality}

Estimating behavioral parameters jointly with stochastic network equilibrium has long been recognized as a challenging task.
A common approach is to formulate parameter estimation as a bilevel problem, with an estimation criterion in the upper level and network equilibrium in the lower level.
For example, \citet{boyce2003validation} estimate parameters by maximizing likelihood while enforcing user equilibrium.
Such formulations are generally non-convex once behavioral parameters and equilibrium flows are estimated jointly, and maximum-likelihood estimation can become computationally burdensome in networks with many OD pairs \citep{sheffi1985urban}.

A second approach fits aggregate network observations, particularly link counts, by minimizing the discrepancy between observed and model-implied flows.
\citet{yang2001simultaneous}, for instance, estimate OD demand and a travel-cost coefficient subject to logit-based stochastic user equilibrium (SUE) conditions.
Their equilibrium-constrained estimation problem is formulated as a single-level, continuously differentiable nonlinear program. Because the travel-cost coefficient, OD demand, and equilibrium link flows are coupled through nonlinear SUE constraints, the resulting simultaneous estimation problem is generally non-convex. They therefore develop a successive quadratic programming algorithm to obtain a local KKT solution.
Moreover, their solution method relies on derivatives of the stochastic network-loading procedure, making its direct extension to a hierarchical destination--mode--route setting with multiple modes nontrivial.

Our approach adopts a different estimation architecture.
Rather than fitting observed travel patterns directly in the objective while imposing equilibrium conditions as constraints, we encode empirical information through moment and conditional-entropy constraints and optimize over the travel distribution.
This primal--dual construction separates equilibrium characterization from the recovery of a broad class of behavioral parameters.
The primal solution characterizes the hierarchical destination--mode--route equilibrium, while the dual variables associated with the empirical constraints recover taste coefficients, alternative-specific constants, and destination- and mode-level scale parameters.
Because these parameters do not appear as primal decision variables, the formulation avoids the bilinear interactions that otherwise arise in joint equilibrium estimation.

The resulting estimation problem is convex conditional on the route-level dispersion parameter $\hat{\lambda}$. This parameter governs how OD flows are distributed across routes and mapped onto link flows, and jointly estimating it with endogenous network flows would reintroduce non-convexity. Rather than treating the entire joint estimation problem as inherently non-convex, our formulation identifies and isolates this remaining source of non-convexity. We therefore calibrate $\hat{\lambda}$ against observed link counts through a one-dimensional outer search, while retaining the recovery of the remaining high-dimensional behavioral parameters within a globally solvable convex program. In this way, the non-convex component is confined to a scalar calibration problem without sacrificing the convex structure of the overall estimation framework.

Our formulation builds on two well-established ideas: the primal--dual relationship between entropy maximization and maximum likelihood estimation in exponential-family models \citep[\S~3.6]{wainwright2008graphical}, and modern exponential-cone representations for entropy and Beckmann terms \citep{mosek_conic}.
By combining these ideas with hierarchical choice and congestion-dependent network equilibrium, we obtain a tractable primal--dual framework for simultaneous equilibrium characterization and parameter recovery.
}

\subsection{Contributions}

In summary, this paper offers significant contributions in three key areas.
\begin{itemize}[leftmargin=2em]
    \item We propose a two-stage travel demand forecasting framework consisting of parameter estimation and future prediction. \HL{In the estimation stage, empirical information is incorporated through moment and conditional-entropy constraints, whose dual variables recover the corresponding behavioral parameters.}
    \item Within this framework, we incorporate flexible choice models that extend the basic multinomial logit structure to capture richer behavioral patterns across destination, mode, and route choices. In particular, the nested logit model accounts for correlations among modes, while the path-size logit model accounts for route overlap.
    \item We develop compact convex formulations and exponential cone reformulations for both the first- and second-stage problems. These conic formulations enable the use of state-of-the-art conic solvers, and our numerical experiments demonstrate computational scalability across benchmark networks of varying sizes.
\end{itemize}

The remaining paper is organized as follows. 
\YS{In \cref{section:modelingassumption}, we describe our modeling assumption. Specifically, we introduce the hierarchical extended logit model.
In \cref{section:twostage}, we formulate our novel two-stage approach for travel demand forecasting.}
\Cref{sec:results} presents numerical results and sensitivity analyses on benchmark networks, demonstrating the behavioral properties and computational scalability of our framework.
Lastly, in \Cref{sec:conclusions}, we discuss the implications of our research and suggest future avenues for advancing travel demand forecasting and transportation planning.
\YS{All proofs are provided in Appendix A, while Appendix B offers a more detailed background on the sequential modeling approach.}

\subsection{Notations} \label{section:notations}
We define the indices: origin $i$, destination $j$, mode $m$ in nest $\nest$, and arc $a$. 
The calligraphic letters \iset~ and \jset~ are reserved for the set of origins and the set of destinations, respectively.
We denote by \msetfull~the set of all modes available for travel from origin~$i$ to destination ~$j$. 
We further partition \msetfull~ into several subsets, which we call \emph{nests}. 
We denote a nest of \msetfull~ by $\nest \in \nestsetfull$, where $\Pi$ denotes partition.
Furthermore, we use \rsetfull~ to denote the set of routes from origin $i$ to destination $j$ using transportation mode $m$. Finally, we use \asetfull~ to denote the set of arcs (links) in route $r$ when traveling from origin $i$ to destination $j$ using transportation mode $m$.
With slight abuse of notation, we sometimes use \mset, \nestset, \rset, and \aset~ without specifying their dependence on their arguments, i.e., the specific origin, destination, mode, or route. For example, the summation $\sum_{i \in \iset} \sum_{j \in \jset} \sum_{m \in \msetfull} x_{ijm}$ is simplified to $\sum_{i \in \iset} \sum_{j \in \jset} \sum_{m \in \mset} x_{ijm}$.

\section{Modeling Assumptions}\label{section:modelingassumption}
We begin by formalizing our behavioral framework, where each traveler is assumed to choose a destination, mode, and route to maximize their individual utility. These choices, while made simultaneously by the traveler, are captured within a hierarchical extended logit model, as illustrated in \cref{fig:hierarchical}. This structure allows us to model complex interactions---for instance, how congestion at the route level influences destination choice---within a single unified framework.  We will then demonstrate that the aggregation of these individual utility-maximizing choices corresponds to the optimal solution of a network-level equilibrium problem.

In this hierarchical model, each level of the decision hierarchy is represented by a logit model specifically chosen for its behavioral realism. Route choice is modeled with a Path Size Logit (PSL) model, which explicitly handles the route overlapping problem \citep{ben1999discrete}. Mode choice is modeled with a Nested Logit (NL) model to account for correlations between similar transport modes. Destination choice is modeled with a Multinomial Logit (MNL) model, which is equivalent to the convex formulation of the extended gravity model (see Appendix~\ref{section:preliminary}).
Formally, the hierarchical extended logit model is parameterized by the destination, mode, and route parameters~\scalej,~\scalem,~\scaler~, respectively. 
In this model, the total utility received from a single trip from origin~$i$ to destination~$j$ on mode~$m$ and route~$r$, denoted by $U_{ijmr}$, can be written as
\begin{align*}
    U_{ijmr} = \cutilj + \cutilm + \cutilr + E_{j|i} + E_{m|ij} + E_{r|ijm},
\end{align*}
where \cutilj, \cutilm, \cutilr~ are the fixed utilities for a traveler choosing destination $j$, mode $m$, and route $r$, and the random additive noises $E_{j|i}, E_{m|ij}$, and $E_{r|ijm}$ follow the cumulative distribution functions defined as  
\begin{align*}
    \begin{cases}
        F(\errj; \scalej) = \exp({- \exp({- \scalej \errj})}), \\
        F(\errm; \scalem) = \exp({-\sum_{\nest \in \nestset} [\sum_{m \in \nest} \exp({-\scalem \varepsilon_m/{\hat{\tau}_{\nest}}}) ]^{\hat{\tau}_{\nest}}}), ~\text{and} \\
        F(\errr; \scaler) = \exp({- \exp({- \scaler \errr})}),
    \end{cases}
\end{align*}
respectively. 
While \errj~ and \errr~ are defined as the Gumbel distribution, \errm~ embeds the correlation within nests. 
\YS{Since the error terms follow the Gumbel distribution, the resulting equilibrium flow distribution can be interpreted as an outcome optimized over the Gumbel-like distribution with the additional considerations regarding Beckmann equation.}
Note that the parameters \scalej,~\scalem,~\scaler~ are related to the size of the variances of their respective distributions, and they appear as scaling factors for each hierarchical structure.
The scaling factors allow us to analyze different levels of choice (i.e., destination, mode, and route), each with a different scale of utility, using a single model.
Since the correlation at the upper level is stronger than at the lower level, the variance and the scale parameter are smaller at the upper level.
From now on, we set \scaler~ to 1. Subsequently, \scalej~ and \scalem~ will be always smaller than 1.

\begin{figure}[!tb]
    \centering
    \includegraphics[width=0.6\textwidth]{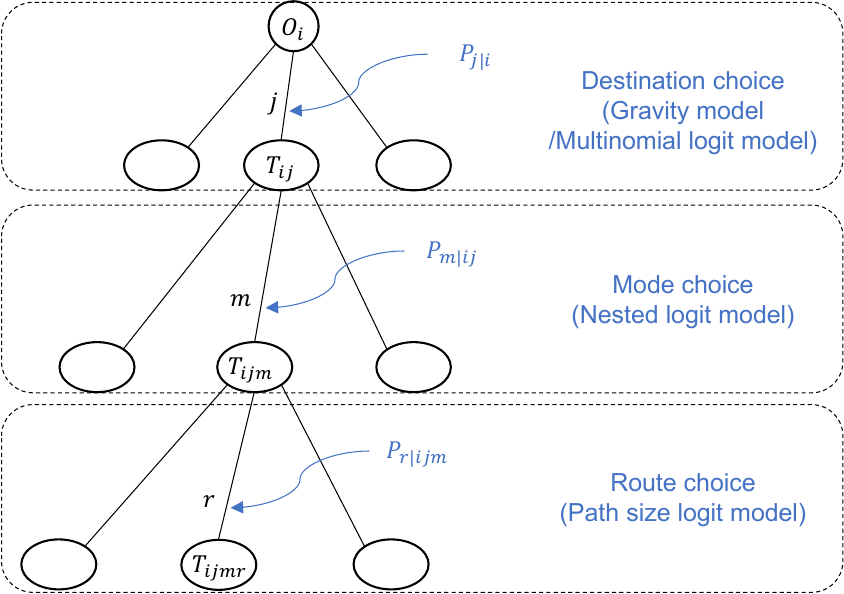}
    \caption{Hierarchical decision-making choices for destination $j$, mode $m$, and route $r$. $P_{j|i}, P_{m|ij}$, and $P_{r|ijm}$ denote the conditional probabilities for destination, mode, and route choices, respectively.}
    \label{fig:hierarchical}
\end{figure}

\subsection{Path size logit model for route choice}
We define the conditional probabilities from the lowest levels. First, we need to define the path-size scaling factor to account for overlapped routes, as certain network arcs may be shared by different routes. The path-size scaling factor is defined as 
\begin{align*}
    \pathsizepar &= \sum_{a \in \asetfull} \frac{\hat{l}_{a}}{\widehat{L}_{ijmr}} \left( \frac{1}{\sum_{r' \in \rsetfull} \widehat{\delta}^{a}_{ijmr'}} \right), 
\end{align*}
where $a$ denotes an arc in the route $r$, which is specified by mode $m$, $\hat{l}_{a}$ denotes the length of link $a$, $\widehat{L}_{ijmr}$ denotes the length of route $r$, and \ind~ denotes the link-route incidence indicator. In this way, using the path-size scaling factor \pathsizepar, each link $a$ in route $r$ is penalized according to the number of routes that share the link $a$. The significance of penalization is weighted by the dominance of link $a$ in route $r$, measured as its length relative to the total route length ($\hat{l}_{a}/\widehat{L}_{ijmr}$).
\YS{Recall that \rsetfull~ denotes the set of routes from origin $i$ to destination $j$ using transportation mode $m$. 
Note that \asetfull~ is a pregenerated candidate set of routes, and the appropriate number of routes can be determined through sensitivity analysis in practice. 
As previous literature has shown, the choice of candidate routes affects the outcome of the equilibrium model and must be carefully designed in real-world contexts \citep{bekhor2008effects, bliemer2008impact}.
Details regarding the generation of the candidate route set are provided in the experimental section.}

The total travel cost of taking route $r$ on mode $m$ from origin $i$ to destination $j$, $g_{ijmr}$, is defined as the summation of the travel cost of all arcs (i.e., roads) in the route and is expressed as 
\begin{align*} \textstyle
g_{ijmr} &= \sum_{a \in \asetfull}  \latency  (\arcflow) \widehat{\delta}^{a}_{ijmr},
\end{align*}
where \latency~ is called the road latency function and is used to measure travel time on congested roads. It takes the flow count \arcflow\ as input of each arc $a$ in the network for mode $m$, and it is typically a convex increasing function.

The generalized utility of taking route $r$ on mode $m$ from origin $i$ to destination $j$ is defined as the scaled negative route cost penalized by the path-size scaling factor and is expressed as
\begin{align} \textstyle
\cutilr &= - \scalepar g_{ijmr}  + \ln(\pathsizepar), 
\label{eqn:utilroute}
\end{align}
where $\scalepar$ \HL{is a scaling parameter governing travelers’ sensitivity to route costs.}

Recall $\mathcal R_{ijm}$ is the set of all routes from origin $i$ to destination $j$ via mode $m$. The probability of choosing route $r$ from the set $\mathcal R_{ijm}$ depends on the fixed utility \cutilr~ through
\begin{align}\label{eqn:c_prob_route}
    \cprobr 
    &= \frac{ \exp({ \cutilr })}{\sum_{r' \in \mathcal R_{ijm}} \exp({ V_{r'|ijm} })}.
\end{align}

\subsection{Nested logit model for mode choice}

\paragraph{Basics of nested logit model.}
The MNL model is a widely used choice model that assumes the independence of irrelevant alternatives (IIA), meaning the introduction or removal of an alternative should not impact the relative preference between other alternatives. 
Because of this assumption, the MNL has a significant limitation: it often fails to capture realistic substitution patterns among alternatives. 
This is evident in scenarios such as the well-known \textit{blue bus and red bus} paradox~\citep{mcfadden1980econometric}. 
Imagine a transportation market featuring private cars and red buses, each with a 50\% market share. 
Now, a new alternative is introduced -- a blue bus. 
According to the IIA assumption, the blue buses should take shares equally from both cars and red buses. 
For example, the market share for each mode could end up being 33.3\%. 
However, this prediction seems unreasonable because the introduction of the new bus disproportionately impacts the market share of existing alternatives. 
Intuitively, the blue buses would primarily take market share from the red buses rather than from private cars.
Therefore, it is more reasonable to expect that the total bus share would remain close to 50\%. 

To address this limitation and better reflect the complexities of decision-making, researchers often resort to introducing the NL. 
The NL model relaxes the IIA assumption by allowing for nested groups of alternatives, enabling a more flexible representation of substitution patterns. 
By incorporating the NL framework, analysts can more accurately model consumer choices in situations with multiple alternatives and better understand how the introduction of new options influences existing choices.

% \begin{figure}[H]
%     \centering
%     \includegraphics[width=\textwidth]{figures/nested.pdf}
%     \caption{Nested logit model}
%     \label{fig:nestedlogit}
% \end{figure}

\begin{figure}[H]
    \centering
    \begin{subfigure}[b]{0.59\textwidth}
        \includegraphics[width=\textwidth]{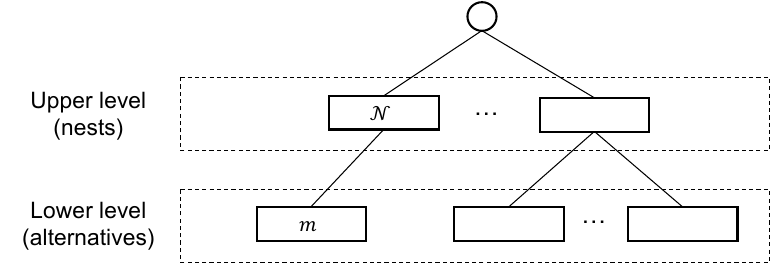}
        \caption{Structure of nested logit model}
        \label{fig:sub:1}
    \end{subfigure}
    \hfill
    \begin{subfigure}[b]{0.39\textwidth}
        \includegraphics[width=\textwidth]{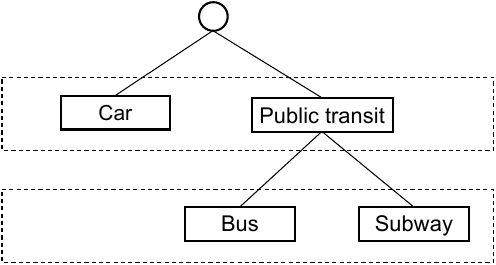}
        \caption{Example of nested mode choice behavior}
        \label{fig:sub:2}
    \end{subfigure}
    \caption{Structure and example of the nested logit model.}
    \label{fig:nestedlogit}
\end{figure}

The NL model is a generalization of the MNL model that accounts for correlations among alternatives within nested groups~\citep{wen2001generalized}. 
We consider a 2-layer model as presented in \cref{fig:nestedlogit}, but this can easily be extended to~$n$-layer NL models. The main difference between the NL and MNL lies in the existence of nests that represent similar alternatives within a group. 

Recall that \msetfull~ denotes the set of all alternatives (modes) for travel from origin $i$ to destination $j$, and we will remove the subscripts $ij$ for simplicity. We use $m$ to denote an element of \mset.
Recall also that \mset~ was partitioned into several nests \nestset, where each nest represents a subset of modes that share similar characteristics or are perceived to be more similar by travelers. We denote a nest by $\nest \in \nestset$. 
The NL model introduces the concept of similarity among alternatives within a nest. 
We denote the dissimilarity factor of the nest \nest~ by \nestpar, which can be expressed by~$\sqrt{1-\hat{\rho}_{\nest}}$ where~$\hat{\rho}_{\nest}$ is the coefficient of correlation. 
When there is no correlation (i.e.,~$\hat{\rho}_{\nest} = 0$), the dissimilarity factor has the highest value (i.e.,~$\hat{\tau}_{\nest} = 1$), which is in the special case that the NL model equals the MNL model.

We denote the utility value of the alternative $m$ by $U_m$. In the random utility model, the utility values are random and follow the additive model $U_m = V_{m} + E_{m}$, where $V_m$ denotes the deterministic values and $E_m$ denotes the random additive noise. The noise distribution in the NL model follows the cumulative distribution function embedding the correlation within nests.  
\begin{align}
    \textstyle F([\varepsilon_m]_{m \in \nest, \nest \in \nestset}; \theta, [\hat{\tau}_{\nest}]_{\nest \in \nestset}) = \exp({-\sum_{\nest \in \nestset} [\sum_{m \in \nest} \exp({-\theta \varepsilon_m/{\hat{\tau}_{\nest}}})]^{\hat{\tau}_{\nest}}}). 
    \label{eqn:NL_cdf}
\end{align}

We denote the probability of choosing alternative $m$ from the nest \nest~ by $p_{m |\nest}$ using the expression
\begin{align*}
    p_{m | \nest} = \frac{\exp({\theta V_{m}/{\hat{\tau}_{\nest}}})}{\sum_{m' \in \nest}\exp({\theta V_{m'}/\hat{\tau}_{\nest}})},
\end{align*}
which is a logit model with the utility term scaled by~$\hat{\tau}_{\nest}$. 
Furthermore, we denote the probability of choosing the nest \nest~ by $p_{\nest}$, which satisfies 
\begin{align*}
    & p_\nest = \frac{\exp({IV_{\nest}})}{\sum_{\nest \in \nestset}\exp({IV_{\nest}})}, \text{  where  } \textstyle IV_{\nest} = \hat{\tau}_\nest \ln \sum_{m' \in \nest }\exp({\theta V_{m'}/\hat{\tau}_{\nest}}). 
\end{align*}    
In this way, the probability of choosing nest~$\nest$ is expressed by a logit model with the inclusive value $IV_{\nest}$, which aggregates the values of alternatives in the nest.
The inclusive value~$IV_{\nest}$ represents the attractiveness of nest~$\nest$ since it measures the summation of normalized utilities of all alternatives in nest~$\nest$, which allows us to compare alternatives across nests. 
Finally, we denote the probability of choosing alternative~$m$ by~$p_m$. If $m$ belongs to the nest \nest, we can compute $p_m = p_{\nest} p_{m|\nest}$ using the Bayes rule, which is also closely related to the problem of entropy maximization derived for the NL model by \citet{fosgerau2020discrete}.
\\

\paragraph{Nested logit model in our context.}
The probability of choosing mode $m \in \nest$ from origin $i$ to destination $j$ satisfies
\begin{align}\label{eqn:c_prob_mode}
    \cprobm &= \cprobnest \cdot \cprobmnest
    = \frac{\exp({IV_{\nest}})}{\sum_{\nest' \in \nestset} \exp({IV_{\nest'}})} \cdot
    \frac{\exp({ \scalem (\cutilm + \satistripm )/\nestpar})}{\sum_{m' \in \nest} \exp({\scalem (V_{m'|ij} + S_{ijm'})/\nestpar})},
\end{align}

where the inclusive value representing the attractiveness of nest $\nest$ is defined as 
\begin{align*}
    IV_{\nest} &=  \ln \left( \textstyle \sum_{m \in \nest} \exp({\scalem [\cutilm + S_{ijm}] /  \nestpar}) \right),
\end{align*}

the satisfaction function $S_{ijm}$ representing the expected utility of choosing mode $m$ is defined as 
\begin{align*}
    \satistripm &= \ln \left( \textstyle \sum_{r \in \rsetfull} \exp({ \cutilr}) \right),
\end{align*} 
and the fixed utility is defined as 
\begin{align} \textstyle
    \cutilm = \sum_{q \in \qset} \tripmattrpar \tripmattr. 
    \label{eqn:utilmode}
\end{align}
Here, \tripmattr~ is the $q$-th travel attribute and \tripmattrpar~ is the parameter associated with the $q$-th attribute, and $\qset$ is the set of all travel attributes \FR{(e.g., travel time, travel cost)} in the mode choice level.

\subsection{Multinomial logit model for destination choice}

The probability of choosing destination $j$ from origin $i$ is computed as
\begin{align}\label{eqn:c_prob_destin}
    \cprobj = \frac{\exp({ \scalej ( \cutilj + \satistrip)})}{\sum_{j' \in \jset} \exp({\scalej (V_{j'|i}+S_{ij'}) })}, 
\end{align}
where the fixed utility is defined as 
\begin{align} \textstyle
    \cutilj = \sum_{k \in \kset} \tripattrpar \tripattr, 
    \label{eqn:utildest}
\end{align}

where \tripattr~ is the $k$-th travel attributes, \tripattrpar~ is the parameter associated with the attribute, and \kset~ is the set of all destination attractiveness attributes \FR{(e.g., employment, point of interest (POI) density)} in the destination choice level, and the satisfaction function is defined as
\begin{align*}
    \satistrip &= \ln \left( \textstyle \sum_{\nest \in \nestsetfull} \exp({ IV_{\nest}}) \right).
\end{align*}

In this section, we assume that the travel demand distribution follows a hierarchical extended logit model. 
In the first stage, we are only given partial observation of travel demand data, and we aim to estimate the parameters of a hierarchical extended logit model by solving a convex optimization problem. 
We then use the estimated parameters to predict the future travel demand given a new set of demand data in each origin.  

\section{Proposed Two-Stage Model}\label{section:twostage}
\YS{We present a two-stage framework for travel demand calibration and prediction. Our approach follows an ``estimate-then-predict'' structure: the first stage uses available data to estimate key behavioral parameters, while the second stage uses these parameters to predict future demand distributions for new scenarios. This structure significantly reduces the internal discrepancies found in traditional sequential models (see \cref{fig:comparison}).
In the following, we will detail this framework, beginning with the model's inputs and data considerations. We then present the mathematical formulation for the first-stage estimation problem, followed by the second-stage prediction problem. Key notations used throughout are defined in \cref{section:notations}.}

\YS{\subsection{Data acquisition}}

\YS{We use the notation $\hat{\cdot}$ to denote given quantities.}
Let \obstotal~ denote the total number of observed travelers, \obsori~ denote the total number of travelers from origin $i$, \obstrip~ denote the number of travelers from origin $i$ to destination $j$, and \obstripm~ denote the number of travelers from origin $i$ to destination $j$ using specific mode $m$. 
\YS{Let $\hat{f}^m_a$ denote the observed number of travelers on link $a$ in the transportation network using mode $m$. }

\HL{
We leverage two primary types of data sources: (1) mode-specific OD matrices, denoted as $\obstripm$, and (2) traffic count data collected from sensors $\hat{f}^m_a$, where $m$ denotes automobiles, such as traffic cameras (e.g., I-24 Motion \citep{i24motion}) or loop detectors (e.g., the Performance Measurement System, PeMS \citep{PEMS}).
Traffic count data are generally considered reliable, as they capture the total number of vehicles traversing a given corridor.
The empirical aggregate travel distribution ($\obstripm$) can be constructed from either directly observed or reconstructed travel patterns.

First, mode-specific OD flows may be obtained directly for selected modes.
For example, automated fare collection systems for public transit (e.g., smart cards or single-use magnetic fare tickets) can capture precise entry and exit information at fare gates.
Similarly, detailed taxi trip records may be available through government open-data portals or data-sharing agreements with taxi operators.
In practice, partially disaggregated data are also common: aggregate OD matrices may be publicly available, while mode-specific information is reported only as aggregate mode shares.
Our formulation can accommodate such cases with minor modifications by matching destination-level OD information and mode-level aggregate statistics separately, rather than requiring a fully disaggregated mode-specific OD matrix.

Second, mode-specific OD flows can be reconstructed by integrating heterogeneous data sources, including household travel surveys, mobile phone and location-based data, and census data \citep{FHWA2022OD}.
Activity-based or agent-based models may also be used to integrate these heterogeneous data sources and synthesize a detailed baseline representation of current travel patterns.
Our framework can then use this reconstructed distribution as an input to recover an interpretable set of behavioral parameters while accounting for network topology and congestion.
In this sense, the two approaches are complementary: activity-based and agent-based models offer substantial flexibility in representing the current system, whereas our formulation provides a compact structure for characterizing the implied behavior and transferring the estimated parameters to future network and demand scenarios.}

\TR{While previous work such as \cite{yang2001simultaneous} focuses on automobile traffic in a road-network setting, in urban travel patterns involving multiple modes and destination choice, link counts alone are insufficient to characterize destination, mode, and route choices or to identify the associated behavioral parameters.
Our modeling philosophy differs from much of the earlier literature, which primarily minimizes discrepancies between observed and simulated link counts.
In contrast, our approach is grounded in theory-driven models (e.g., entropy-based representations of logit models and Beckmann equations), while using link counts to calibrate the route-level dispersion parameter and to assess consistency with observed network usage.} %as additional information to balance observed data with theoretical equilibrium or

\subsection{First stage – estimation}\label{section:firststage}
\HL{The observed data provide only partial information about the full travel distribution. Mode-specific OD matrices characterize aggregate destination and mode choices, but do not reveal how trips are distributed across routes. Conversely, traffic sensors (e.g., loop detectors or magnetometers) provide vehicle counts at monitored roadway locations, which can be aggregated to link-level flows, without directly identifying the corresponding origins, destinations, or routes. Consequently, the complete joint distribution over destination, mode, and route choices remains unobserved.}
Traditional maximum likelihood estimation is commonly used to estimate parameters from such incomplete data. \citet[\S~3.6]{wainwright2008graphical} demonstrated that the dual of the entropy maximization problem is equivalent to maximum likelihood estimation in exponential-family models. Furthermore, \citet{donoso2011maximum} showed that the constrained maximum entropy provides a better estimate than maximum likelihood, as it reproduces both the average values of explanatory variables and the observed market modal shares.

\HL{Building on these findings, we consider a constrained entropy maximization approach for the first-stage estimation problem. 
The primal problem first converts the available aggregate observations into an empirical distribution. It then searches for a distribution that is consistent with these observations and maximizes the entropy-based objective.
Intuitively, this selects a distribution that remains as dispersed as possible while satisfying the empirical information imposed through conditional-entropy and moment constraints. 
However, the behavioral interpretation of this objective may not be immediately intuitive from the primal formulation alone.
Its meaning becomes clear through the optimality conditions: the resulting primal probabilities satisfy the hierarchical extended logit model, while the associated dual variables recover the corresponding behavioral parameters.}

We start by formulating the primal problem. 
\YS{From the available data—\obstotal, \obsori, \obstrip, and \obstripm}—we first construct the empirical probability distributions. 
Specifically, the joint probability of departing from origin $i$, the probability of choosing origin-destination pair $(i,j)$, and the probability of choosing origin-destination-mode triplet $(i,j,m)$, denoted by \obsprobi, \obsjprobj, and \obsjprobm, respectively, are defined as
\begin{align*} \textstyle
    \obsprobi = {\obsori} /{\obstotal} , \quad
    \obsjprobj = {\obstrip}/{\obstotal}, \quad \text{and} \quad \obsjprobm = {\obstripm}/{\obstotal}.
\end{align*}

The empirical conditional probabilities are defined as 
\begin{align*} \textstyle
    \obscprobj = {\obsjprobj}/{\obsprobi} , \quad 
    \obscprobmnest = {\obsjprobm}/{\obsjprobnest} , \quad \text{and} \quad
    \obscprobnest = {\obsjprobnest}/{\obsjprobj}.
\end{align*}

To represent the vector for joint probabilities of destination choice $[\jprobj]_{\forall i \in \iset, \forall j \in \jset}$, we employ the compact notation $\jprobvecj$. 
Similarly, we use $\jprobvecm$, $\jprobvecgivennest$, $\jprobvecnest$, and $\jprobvecr$ to denote the vectors of joint probabilities for modes, a specific nest, nests, and routes, respectively. 
We define the vectors of conditional probabilities for destination, mode, and route choice $\cprobvecj = \jprobvecj|\obsprobivec$, $\cprobvecm = \jprobvecm|\jprobvecj$, and $\cprobvecr = \jprobvecr|\jprobvecm$, respectively.

In the following, we define conditional entropy functions for each of destination, mode, and route choice models \citep{fosgerau2022inverse, gomez2022optimal, bekhor2001stochastic}:

\begin{align}
    H_{\text{MNL}}(\cprobvecj) & = - \sum_{i \in \iset, j \in \jset} \jprobj \ln \left(\frac{\jprobj}{\obsprobi} \right), \\
    H_{\text{NL}}(\cprobvecm) & = H_{\text{NL}_1} \left( \jprobvecnest|\jprobvecj \right) + \sum_{\nest \in \nestset} \hat{\tau}_{\nestofm} H_{\text{NL}_{\nestofm}}\left( \jprobvecm|\jprobvecgivennest \right) \\ 
    & = - \sum_{i \in \iset, j \in \jset} \sum_{\nest \in \nestset} \jprobnest \ln \left( \frac{\jprobnest}{\jprobj} \right) - \sum_{\nest \in \nestset} \hat{\tau}_{\nestofm} \sum_{i \in \iset, j \in \jset}
    \sum_{m \in \nestofm} \jprobm \ln \left( \frac{\jprobm}{p_{ij\nestofm}} \right)  \nonumber, \\
    H_{\text{PSL}}(\cprobvecr) & = - \sum_{i \in \iset, j \in \jset} \sum_{m \in \mset, r \in \rset} \jprobr \ln \left( \frac{\jprobr}{\jprobm} \cdot \frac{1}{\pathsizepar} \right).
\end{align}

\citet{beckmann1956studies} showed that minimizing the sum of integrals of all road latency functions, subject to certain constraints, characterizes a user equilibrium according to Wardrop's principle. This principle requires that travel times on all used routes are equal \citep{beckmann1956studies, sheffi1985urban}. 
The Beckmann equation is defined as 

\begin{equation} \label{eqn:beckmann}
    \beckmann.
\end{equation}

Recall that \latency~ is a function representing the road latency, taking the flow count \arcflow~ as an input. 
By incorporating the Beckmann equation, we can consider the endogenous travel time, which depends on arc-level flows and congestion.

%\SSC{why is lambda in red?}
Our goal is to jointly maximize the entropy functions and minimize the Beckmann equation. The objective function is to maximize 
\begin{align} 
    &  H_{\text{MNL}}(\cprobvecj) + H_{\text{NL}}(\cprobvecm) + H_{\text{PSL}}(\cprobvecr)  - \frac{\scalepar}{\obstotal}\beckmann. \label{FirstStage:objective}
\end{align}

The inclusion of entropy terms alongside the Beckmann equation means our optimal solution characterizes a user equilibrium that is different from Wardrop's principle of user equilibrium. 
Instead, it adheres to the hierarchical extended logit model, which will be demonstrated later in \cref{theorem:convex}.
On the other hand, compared to models without the Beckmann equation, incorporating it enables characterizing the stochastic user equilibrium by accounting for congestion impacts. 
The congestion in each arc can be effectively captured by its traffic count $\arcflow$ and latency function $\latency$. 

\HL{
Recall that the parameter $\scalepar$ was introduced in \cref{eqn:utilroute} to represent travelers’ sensitivity to route costs.
Mathematically, multiplying the Beckmann term by $\scalepar$ ensures that the first-order conditions are consistent with the route utility specification in \cref{eqn:utilroute}. Equivalently, in the optimization objective, $\scalepar$ adjusts the relative weight of the Beckmann term compared with the entropy terms.
Under this interpretation, a larger value of $\scalepar$ places greater emphasis on route costs, producing more concentrated and less dispersed route choices. Hence, $\scalepar$ can also be interpreted as an inverse-dispersion parameter that governs how O-D flows are distributed across routes and, consequently, mapped onto link flows.
}

\HL{
We treat $\scalepar$ as a structural calibration parameter that controls the relative scaling of the Beckmann and entropy terms. 
In practice, we calibrate $\scalepar$ by minimizing the discrepancy between predicted and observed road link counts. 
Specifically, we use the normalized sum of squared errors
\begin{align}\label{metric}
\argmin_{\lambda >0} =
\frac{1}{\obstotal}
\sum_{a \in \aset_{\mathrm{obs}}}
\left(
f^{\mathrm{car}}_a (\lambda) -\hat{f}^{\mathrm{car}}_a
\right)^2,
\end{align}
where $\aset_{\mathrm{obs}}\subseteq\aset$ denotes the set of links with available traffic-count observations, such as those obtained from loop detectors \citep{PEMS}. 
Because $\lambda$ is scalar, this calibration can be performed through a simple one-dimensional search. 
Conditional on the selected value of $\scalepar$, the remaining behavioral parameters, including taste coefficients, alternative-specific constants, and destination- and mode-level scale parameters, are recovered simultaneously from the dual solution of a single convex program.
}

The entropy function and Beckmann equation are connected through the relationship between the joint route choice probabilities $\jprobr$ and arc flows $\arcflow$, given by
$$\arcflow = \sum_{i \in \iset, j \in \jset} \sum_{r \in \rsetfull} \obstotal \jprobr \ind, $$
where the indicator function $\ind$ specifies whether arc $a$ is included in route $r$, taking a binary value of 0 or 1 as a predefined known integer.

Let $\Pset$ be a set $\{ (\jprobj, \jprobnest, \jprobm, \jprobr, 
\arcflow)_{{i \in \iset, j \in \jset, 
\nest \in \nestset, m \in \mset, r \in \rset, a \in \aset }}\} $ such that 

\begin{align}
    & \sum_{j \in \jset} \jprobj = \obsprobi &,\forall i \in \iset & \quad \label{FirstStage:prob1}\\ 
    & \jprobj = \sum_{\nest \in \nestset} \jprobnest &,\forall i \in \iset ,j \in \jset & \quad \label{FirstStage:prob2}\\
    & \jprobnest = \sum_{m \in \nest} \jprobm &,\forall i \in \iset , j \in \jset,\nest \in 
    \nestset& \quad \label{FirstStage:prob3}\\
    & \jprobm = \sum_{r \in \rsetfull} \jprobr &,\forall i \in \iset , j \in \jset, m \in \mset & \quad \label{FirstStage:prob4} \\
    & \arcflow = \sum_{i \in \iset, j \in \jset} \sum_{r \in \rsetfull} \obstotal \jprobr \ind & , \forall a \in \aset, m \in \mset 
    \label{FirstStage:flow} \\
    & \jprobj, \jprobnest, \jprobm, \jprobr \geq 0 & , \forall i \in \iset , j \in \jset, m \in \mset, r \in \rset \label{FirstStage:positive}
\end{align}

Constraints in \cref{FirstStage:prob1}–\cref{FirstStage:prob4} link the joint probabilities at each step. \cref{FirstStage:flow} shows the relationship between link flow and route flow using the indicator $\ind$.

% \SSC{The first state problem is listed without any introduction.}

We now introduce the First Stage model to characterize an equilibrium of current travel patterns and to estimate parameters. 

\begin{center}
\fbox{\parbox{0.9\columnwidth}{ {\centering
\vspace{0.5ex}\textsc{First Stage Problem}\\[1ex]}
		
\begin{align*}
    \max_{\Pset} & \text{ \cref{FirstStage:objective}} 
\end{align*}

\begin{subequations}
\begin{align}
    \text{s.t.} \nonumber 
    \\ & H_{\text{MNL}}(\jprobvecj|\obsprobivec) \geq H_{\text{MNL}}(\obsjprobvecj|\obsprobivec) && \quad [-1 + \frac{1}{\scalej}]  \label{FirstStage:ratio1} \\
    & H_{\text{NL}}(\jprobvecm|\jprobvecj) \geq H_{\text{NL}}(\obsjprobvecm|\obsjprobvecj) && \quad [-1 + \frac{1}{\scalem}] \label{FirstStage:ratio2}\\
    & \sum_{i \in \iset, j \in \jset} \jprobj \tripattr = \sum_{i \in \iset, j \in \jset} \obsjprobj \tripattr & , \forall k \in \kset & \quad [\tripattrpar] \label{FirstStage:aggregate1} \\
    & \sum_{i \in \iset, j \in \jset} \sum_{ m \in \mset} \jprobm \tripmattr = \sum_{i \in \iset, j \in \jset} \sum_{ m \in \mset} \obsjprobm \tripmattr  & , \forall q \in \qset & \quad [\tripmattrpar] \label{FirstStage:aggregate2}
\end{align}
\end{subequations}
\vspace{0.5ex}
	}}
\end{center}

% \begin{subequations}
% \begin{align}
%     \text{s.t.} \nonumber 
%     \\ & H_{\text{MNL}}(\jprobvecj|\obsprobivec) \geq H_{\text{MNL}}(\obsjprobvecj|\obsprobivec) && \quad [-1 + \frac{1}{\scalej}]  \label{FirstStage:ratio1} \\
%     & H_{\text{NL}_1}(\jprobvecnest|\jprobvecj) \geq H_{\text{NL}_1}(\obsjprobvecnest|\obsjprobvecj) && \quad [-1 + \frac{1}{\scalem}] \label{FirstStage:ratio2}\\
%     & H_{\text{NL}_{\nestofm}}(\jprobvecm|\jprobvecgivennest) \geq H_{\text{NL}_{\nestofm}}(\obsjprobvecm|\obsjprobvecgivennest) &, \forall \nest \in \nestset & \quad [- 1 + \frac{\nestpar}{\scalem}] \label{FirstStage:ratio3}\\
%     & \sum_{i \in \iset, j \in \jset} \jprobj \tripattr = \sum_{i \in \iset, j \in \jset} \obsjprobj \tripattr & , \forall k \in \kset & \quad [\tripattrpar] \label{FirstStage:aggregate1} \\
%     & \sum_{i \in \iset, j \in \jset} \sum_{ m \in \mset} \jprobm \tripmattr = \sum_{i \in \iset, j \in \jset} \sum_{ m \in \mset} \obsjprobm \tripmattr  & , \forall q \in \qset & \quad [\tripmattrpar] \label{FirstStage:aggregate2}
% \end{align}
% \end{subequations}
% \vspace{0.5ex}
% 	}}
% \end{center}

The objective is to maximize \cref{FirstStage:objective} over set $\Pset$.
Constraints in \cref{FirstStage:ratio1} - \cref{FirstStage:ratio2} ensure that the conditional entropy matches the observed conditional entropy. 
Constraints in \cref{FirstStage:aggregate1} and \cref{FirstStage:aggregate2} ensure that the fixed utility and market shares align with observed values, ensuring the moment is matched.
\YS{Recall that \tripmattr~ is the $q$-th travel attribute defined in \cref{eqn:utilmode} and \tripattr~ is the $k$-th travel attributes defined in \cref{eqn:utildest}. 
\qset~ can be viewed as mode-specific attributes for a given origin-destination pair, such as travel time, travel fare, and constants capturing mode-specific preferences (i.e., alternative specific constants).  
\kset~ represents destination-specific attractiveness attributes that are independent of mode, including the number of jobs and point-of-interest (POI) density.}

\HL{Rather than solving the likelihood maximization problem directly, we embed its first-order moment conditions into the entropy-maximization problem as the constraints in \cref{FirstStage:aggregate1} and \cref{FirstStage:aggregate2}. 
These constraints equate the model-implied aggregate attributes with their observed values, while their dual variables, $\tripmattrpar$ and $\tripattrpar$, recover the corresponding taste coefficients of the mode- and destination-choice utility functions. 
Moment-matching constraints have been widely used in maximum-entropy estimation to ensure that model-implied aggregate attributes reproduce their observed values \citet{donoso2011maximum}. 
}

\FR{
Similarly, the dual variables corresponding to \cref{FirstStage:ratio1} - \cref{FirstStage:ratio2} are listed in square brackets as well, $-1 + \frac{1}{\scalej}$ and $-1 + \frac{1}{\scalem}$.
Note that rather than introducing additional auxiliary dual variables, we opted to use the already defined variables (i.e., the scaling parameters $\scalej$ and $\scalem$) as their form should be clear to the reader when deriving the KKT conditions.
}
We make one technical note regarding the boundary conditions for the dual variables $\scalej$ and $\scalem$. 
From the dual feasibility conditions, we have $-1+\frac{1}{\scalej} \geq 0$ and $-1+\frac{1}{\scalem} \geq 0$, which implies that $\scalej \leq 1$ and $\scalem \leq 1$, respectively. 
This aligns with the condition that $\scalem$ and $\scalej$ should be less than 1, as we previously set $\scaler = 1$ and thus omitted it from the formulation.
Recall that the scaling parameters correspond to the variance of the random utility components in the logit model. 
The correlation at the mode and destination levels is appropriately stronger (and thus has a smaller variance) than at the route level, \HL{i.e., $0 < \scalej \leq \scalem \leq \scaler = 1$} \citep{donoso2011maximum, yao2014general}. 

\HL{This primal--dual construction is also the key to preserving convexity during parameter estimation.
If the behavioral parameters were introduced directly as primal decision variables, they would interact nonlinearly with the endogenous probability variables.
For example, taste coefficients would be coupled with endogenous choice probabilities through the systematic-utility terms, while destination- and mode-level scale parameters would additionally interact with the corresponding taste coefficients. By instead recovering these parameters from the dual variables of the empirical moment and conditional-entropy constraints, the primal optimization remains expressed only in terms of equilibrium probabilities and flows, thereby avoiding these parameter--equilibrium couplings.}

\HL{
Our entropy maximization approach with hierarchical probabilities builds on \citet{donoso2011maximum} by incorporating congestion-dependent route costs and stochastic user equilibrium conditions. Whereas \citet{donoso2011maximum} matched the observed conditional entropies through equality constraints, we replace these equalities with lower-bound inequality constraints, as shown in \cref{FirstStage:ratio1}, \cref{FirstStage:ratio2}. This relaxation is important computationally: it preserves convexity and enables an exponential-cone reformulation that can be solved using off-the-shelf conic solvers such as MOSEK, ECOS, and SCS.
The relaxation also retains a useful dual interpretation. The scale parameters associated with destination and mode choice, \(\theta_{\text{dest}}\) and \(\theta_{\text{mode}}\), arise as dual variables of the conditional-entropy constraints, with dual feasibility enforcing their natural upper bound of one. When a conditional-entropy constraint is active, its dual variable yields a nontrivial scale parameter strictly below one. When the constraint is slack, complementary slackness implies that the corresponding scale parameter attains the boundary value of one. Thus, the relaxed formulation does not necessarily yield an interior scale estimate for every dataset; rather, a nontrivial scale estimate arises when the observed entropy bound is informative at the optimum. Together with the moment constraints, this dual structure allows travel-pattern characterization and behavioral-parameter estimation to be carried out within a single convex optimization problem.
}

To introduce our main theorem, we begin by stating the following assumption.

\begin{assumption}\label{assume:latency}
    The road latency function $\latency$ is a strictly increasing convex power function.
\end{assumption}

Under Assumption \ref{assume:latency}, it is well-known that the Beckmann equation in \cref{eqn:beckmann} is convex \citep[see][\S 3.1.]{beckmann1956studies}. 
The Bureau of Public Roads (BPR) function satisfies Assumption \ref{assume:latency}, and is the standard way of quantifying travel time on congested roads. It can be expressed as 
\begin{align}
    \latency(\arcflow) = T^{m,0}_a \left[1+ a^{m} \left(\frac{\arcflow}{c^m_a}\right)^{b^{m}}\right],
\end{align}
where $\latency(\arcflow)$ represents the travel time on a segment (e.g., arc or link in the road network) as a function of flow count on that segment, $T^{m,0}_a$ denotes the free flow travel time when there is no congestion, $c^m_a$ denotes the capacity of the segment, $a^m$ and $b^m$ are parameters calibrated for a particular network. Typically, for the road network (i.e., $m$ is automobiles), the parameters are commonly set to $a^m = 0.15$ and $b^m = 4$. For public transit modes unaffected by traffic (e.g., subways, trams, trains), $b^m = 0$, indicating a constant travel time regardless of the flow.

Under Assumption \ref{assume:latency}, we establish a theorem demonstrating that our formulation possesses the desirable property of convexity. Convexity plays a pivotal role in optimization due to the existence of globally optimal solutions and the availability of efficient algorithms for solving convex problems.

\begin{theorem}\label{theorem:convex}
The \textsc{First Stage Problem} is a convex program under Assumption \ref{assume:latency}. In the \textsc{First Stage Problem}, the optimal primal solutions \jprobj, \jprobnest, \jprobm, \jprobr~ and the optimal dual solutions \scalej, \scalem, \tripattrpar, \tripmattrpar~ satisfy the hierarchical extended logit model defined in \cref{eqn:c_prob_route}, \cref{eqn:c_prob_mode}, \cref{eqn:c_prob_destin} with the fixed utilities in \cref{eqn:utilroute}, \cref{eqn:utilmode}, and \cref{eqn:utildest}.
\end{theorem}

\cref{theorem:convex} demonstrates that the formulation is a convex program. 
The complete proof is provided in Appendix \ref{appendix:hier_extend}.
Later, we will formally prove the existence of an exponential cone reformulation, which is a stronger notion than convexity. 
The convexity allows us to utilize the Karush-Kuhn–Tucker (KKT) conditions, and the theorem highlights that the solution to the \textsc{First Stage Problem} is equivalent to the solution of the hierarchical extended logit model previously defined in \cref{section:modelingassumption}. 
\SR{This theorem clarifies the interpretation of our formulation. While the objective function, as in the Beckmann formulation, may not have a direct physical interpretation, the optimality conditions reveal that its solution corresponds to the hierarchical extended logit model.}

The dual variables of each constraint (listed in square brackets) yield the parameters to be estimated in the hierarchical extended logit model. 
\HL{This parameter-recovery mechanism is closely related to the classical duality between entropy maximization and maximum-likelihood estimation, but the resulting dual problem is not the standard maximum-likelihood estimation problem. In particular, the Beckmann congestion term, the conditional-entropy constraints in \cref{FirstStage:ratio1}–\cref{FirstStage:ratio2}, the additional moment constraints in \cref{FirstStage:aggregate1}–\cref{FirstStage:aggregate2}, and the network-feasibility constraints in \cref{FirstStage:prob1}–\cref{FirstStage:positive} modify the optimization problem beyond conventional likelihood maximization, as shown in Appendix~\ref{appendix:dual}.
This connection aligns with the well-established duality between entropy maximization (i.e., a special case of the proposed formulation) and maximum likelihood estimation \cite[\S~3.6]{wainwright2008graphical}. 
Hence, our work can be viewed as an extension of the framework proposed by \citet{donoso2010microeconomic}, who focused on entropy maximization corresponding to the MNL, to accommodate more flexible logit specifications at each level of the hierarchy and to incorporate additional structural constraints. 
}

\begin{theorem}\label{theorem:expcone}
\textsc{The First Stage Problem} has an exponential conic reformulation under Assumption \ref{assume:latency}. 
\end{theorem}

We can further strengthen our formulation by reformulating the problem as an exponential cone program.
\cref{theorem:expcone} suggests that the First Stage Problem has an exponential cone representation.
To achieve this reformulation, we leverage the fact that the conditional entropy terms can be expressed as exponential cones, and the integral of the latency function remains a power function.
The full proof can be found in \cref{appendix:expo_cone}. 

\YS{
The resulting exponential cone formulation can be natively supported by state-of-the-art conic solvers such as MOSEK. 
This allows the solver to internally transform and solve the problem using exact convex representations, without requiring approximations. 
MOSEK employs an interior-point method \citep{mosek_conic}, which is theoretically guaranteed to converge to a global optimum under Slater’s condition \citep{boyd2004convex}. 
MOSEK computes both primal and dual solutions and reports the associated duality gap, which serves as a numerical certificate of optimality. 
When the duality gap falls within a specified tolerance, the solution is optimal with a certain level of precision.
}

% A more detailed discussion of the solver settings will be provided in the Experimental Results section. 

% \SSC{State the purpose before describing it.}

% Lastly, we suggest an important variant of the First Stage Problem, aiming to strike a balance between our modeling assumptions and observed discrepancies from other data sources. 
% The useful variant is created by adding a term that quantifies the error between the output link flow values $f^{\text{car}}_a$ and the observed values $\widehat{f^{\text{car}}_a}$, to the objective function.
% The observed values $\widehat{f^{\text{car}}_a}$ can be obtained from sources such as loop detector \citep{PEMS}. 

% \TR{This variant remains a convex optimization problem.
% Specifically, the additional term is a quadratic (second-order norm) penalty composed with an affine mapping from the decision variables to link flows, which is a standard convex function.
% Since the original First Stage Problem is convex and the feasible set is unchanged, adding this convex penalty preserves overall convexity of the formulation.}
% By incorporating this additional term, the optimal solution of the program is no longer the travel equilibrium state.
% Since the real world does not exhibit a perfect user equilibrium, we need to balance between the user equilibrium and the observed data. 
% The parameter $\sigma$ governs this trade-off.

\subsection{Second stage – prediction}
The second stage involves predicting future travel distributions using the parameters estimated in the first stage.
Consider the scenario where we aim to estimate new travel patterns 10 years after completing the current large-scale construction of new roads or transit lines.
The future number of travelers from each origin zone \obsori\ can be estimated using regression models that account for demographic changes \citep{mcnally2007four}.
This is usually estimated by a regression model that considers time series patterns of people living in certain origin locations.
Subsequently, the total number of travelers \obstotal~ can be calculated. 
The travel attributes in the changed transportation network, \tripattr~ and \tripmattr, such as travel time and cost, are updated since these values will change after constructing new roads or lines. We use previously estimated parameters for travel attributes $\opttripattrpar$ and $\opttripmattrpar$, scaling factors ~$\optscalej$ and $\optscalem$, and nest parameter ~$\optnestpar$. With these new inputs and previously estimated parameters, we can reconstruct an optimization problem from the Lagrangian function to predict future travel distributions. To express the deterministic utilities as a function of travel attributes and their parameters, as shown in \cref{eqn:utilmode} and \cref{eqn:utildest}, we denote them as $V_{ij} (\{\opttripattrpar, \tripattr\}_{k \in \kset})$
and $V_{ijm} (\{ \opttripmattrpar, \tripmattr\}_{q \in \qset})$.

\begin{center}
\fbox{\parbox{0.825\columnwidth}{ {\centering
\vspace{0.5ex}\textsc{Second Stage Problem}\\[1ex]}

\begin{align*}
    \max_{\Pset} & \sum_{i \in \iset, j \in \jset} \jprobj V_{ij} (\opttripattrpar, \tripattr) + \sum_{i \in \iset, j \in \jset} \sum_{m \in \mset} \jprobm V_{ijm} ( \opttripmattrpar, \tripmattr) \\
    &  + \frac{1}{\optscalej} H_{\text{MNL}}(\cprobvecj) \\
    & + \frac{1}{\optscalem} H_{\text{NL}}(\cprobvecm)    \\
    & + H_{\text{PSL}}(\cprobvecr)    \\ 
    & - \frac{\scalepar}{\obstotal} \beckmann
\end{align*}
	}}
\end{center}

The objective function consists of the summation of deterministic utilities, entropy functions, and the integral of the Beckmann equation.

\begin{corollary}\label{corollary:SecondStage}
    In the Second Stage Problem, the optimal solutions \optjprobj, \optjprobnest, \optjprobm, \optjprobr~ satisfy the hierarchical extended logit model defined in \cref{eqn:c_prob_route} - \cref{eqn:c_prob_destin}. 
\end{corollary}

Under Assumption \ref{assume:latency}, the Second Stage Problem is convex programming with a concave objective function and linear constraints. Corollary \ref{corollary:SecondStage} confirms that the future estimation can be obtained by solving problem {\fontfamily{cmr}\selectfont SecondStage}. 
We omit the proof as it can be immediately derived from the proof of \cref{theorem:convex}.
From the optimal solution \optjprobj, \optjprobnest, \optjprobm, \optjprobr, we recover the full information of future travel patterns.
% \SSC{be consistent with how you refer to the first stage problem}

In summary, the First Stage Problem incorporates observed data from each decision-making step as constraints and simultaneously estimates the associated behavioral parameters through the corresponding dual variables. These estimated parameters are then carried forward to the Second Stage Problem, providing a behaviorally and internally consistent basis for forecasting future travel patterns.

{\subsection{Summary of our contributions}

\HL{We propose a two-stage travel demand forecasting framework that first reconstructs current travel patterns to estimate model parameters and then uses the estimated model to predict future demand. Building on prior single-level formulations of logit-based network equilibrium models, our key methodological contribution is a primal--dual decomposition that separates equilibrium characterization from parameter recovery while preserving convexity for a broad class of behavioral parameters.

In the first-stage problem, the primal variables characterize the joint destination--mode--route equilibrium, whereas the dual variables associated with empirical moment and conditional-entropy constraints recover the taste coefficients, alternative-specific constants, and destination- and mode-level scale parameters. 
This construction avoids introducing these parameters directly as primal decision variables, thereby preventing the bilinear interactions (e.g., coupling between choice probabilities and taste coefficients, and products between destination- or mode-level scale parameters and the corresponding taste coefficients) that would otherwise destroy convexity. 
The resulting formulation therefore extends the classical entropy--likelihood duality to a network equilibrium setting with hierarchical choice structure and congestion-dependent route costs.

The decomposition also makes explicit the boundary of this convex structure. In this sense, our contribution is not to eliminate all sources of non-convexity in joint estimation, but to push the boundary of convex parameter recovery substantially further: taste coefficients, alternative-specific constants, and destination- and mode-level scale parameters are recovered simultaneously within the convex program, leaving only the scalar route-level dispersion parameter outside. Conditional on $\hat{\lambda}$, this high-dimensional parameter-recovery problem remains convex and globally solvable. 
In contrast, jointly estimating $\hat{\lambda}$ with network flows would reintroduce non-convexity because $\hat{\lambda}$ multiplies the flow-dependent Beckmann term, creating a bilinear coupling between the parameter and endogenous equilibrium quantities.
We therefore isolate $\hat{\lambda}$ and calibrate it against observed link counts through a simple one-dimensional outer search. This confines the remaining calibration to a simple scalar search while preserving the convex structure of the high-dimensional estimation problem.

Theorem 1 formalizes this primal--dual relationship and shows that the optimality conditions recover the hierarchical extended logit model. Theorem 2 further establishes an exponential-cone reformulation of the convex subproblem. This conic representation enables the use of state-of-the-art solvers that provide globally optimal primal and dual solutions together with numerical optimality certificates.}

\section{Experimental Results}
\label{sec:results}
We conducted comprehensive experiments exploring various combinations of implementing entropy functions and the Beckmann equation. 
For entropy functions, we tested (a) the exponential cone, (b) the relative cone with dimension 3, and (c) the relative cone with dimension 2n+1. 
For the Beckmann equation, we examined both (x) the power cone and (y) the \gls{a:soc}. 
The \gls{a:socp} reformulation of the Beckmann equation has been suggested by \cite{wei2019efficient}. However, this implementation can encounter numerical problems without proper scaling of the \gls{a:socp} cone. To mitigate these scaling issues, we also explore the power cone, and our observations indeed reveal that the power cone is more efficient than \gls{a:socp}.
We also compared the performance of different solvers, including MOSEK, ECOS, and SCS. Based on our extensive analysis, we determined that the combination of the exponential cone for entropy functions, the power cone for the Beckmann equation, and the MOSEK solver yields the best numerical stability. 
This configuration was applied throughout this section. 

%Although MOSEK can be more efficient for smaller instances, it may face out-of-memory issues with larger instances. Given that our implementation is in Julia, users can easily switch to a first-order solver, such as SCS.

We demonstrate our methodology using several networks. 
First, we analyze a basic test network designed by \citet{yao2014general} to illustrate and underline the route overlapping effects. 
This network is intentionally kept simple to highlight and draw meaningful insights. 
Next, we conduct an in-depth analysis of a stylized Sioux Falls network, used to demonstrate the effect of mode similarity to showcase use cases for the presented methodologies.
Finally, we apply our approach to the Sioux Falls, Anaheim, and Eastern Massachusetts networks to assess its applicability.
The benchmark networks are sourced from \citet{transportnetwork} that are designed to analyze automobile-only traffic assignment. 
However, throughout our experiments, we consider three transport modes: automobile, bus, and subway. 
The provided dataset contains an OD demand matrix, not a mode-specific OD. 

Our experiments primarily focus on the first stage model, as solving the second stage model is relatively straightforward.
In our numerical experiments, we have \cref{FirstStage:ratio1} but not \cref{FirstStage:ratio2}. 
Accordingly, we estimate $\scalej$ and assume that $\scalem$ is set to 1 due to the lack of mode-specific OD matrix. 
Our model utilizes a nested structure for mode choice, placing bus and subway options within a shared public transit nest, as depicted in \cref{fig:nestedlogit}(b). 
In the nested structure, $\tau_{\text{automobile}} = 1$ because the automobile is the only alternative in its own nest. 
We vary the dissimilarity parameter for public transit, $\tau_{\text{public transit}}$, which will be simply denoted as $\tau$ in the later sections. 

All experiment codes and data are available online:  \url{https://github.com/youngseo-Kim/ConicTrafficAssignment}. Experiments were conducted on Ubuntu 18.04.3 LTS Intel Xeon Gold 6244 CPU 3.60GHz using Julia 1.10.0 and Jump interface.

\subsection{A toy network}

To evaluate the impact of route overlapping, we utilize the test network proposed by \citet{yao2014general}. 
In this network, the variable $x$ represents the length of arc $(1 \rightarrow 2)$, which is varied at values of 0, 2, and 3.
Figures (a) and (b) depict the scenario where $x = 0$, causing node \#2 to be overlapped with node \#1 and to disappear. Figures (c) and (d) illustrate the case when the routes are moderately overlapped ($x = 2$), and Figures (e) and (f) represent the situation where the routes are heavily overlapped ($x = 3$).
For clarity and conciseness, we focus our report on the road network with traffic counts for automobiles. 

In the case of no route overlaps, both MNL and PSL yield identical flow distributions as the independent assumption of alternative routes is fully satisfied.
In two route overlapping cases, the MNL model produces the same flow distributions for both cases, while the PSL model generates different flow distributions that explicitly consider the extent of the route overlapping. Essentially, without PSL, MNL overestimates the arc flow on congested roads. This is illustrated by the estimates for the arc $(1 \rightarrow 2)$. MNL estimates 13,095 vehicles, while PSL estimates 12,168 and 11,666, depending on the level of route overlap. In the PSL model, when comparing figures (d) and (f), as $x$ increases from $2$ to $3$, the flow along the arcs in the overlapped routes $(1 \rightarrow 2 \rightarrow 3)$ reduces from 1,645 to 1,480, and similarly, $(1 \rightarrow 2 \rightarrow 4 \rightarrow 3)$ decreases from 1,645 to 1,480.
These trends align with our expectations because, as route overlap increases, the attractiveness of the overlapping routes decreases, leading to a reduction in traffic counts.
It is not surprising to observe the symmetry of flows about the horizontal axis, given that the arc characteristics were deliberately set to be symmetric.

\begin{figure}[!htbp]
    \centering
    \captionsetup[subfigure]{skip=0pt}
    \begin{subfigure}[b]{0.28\textwidth}
        \includegraphics[width=\linewidth]{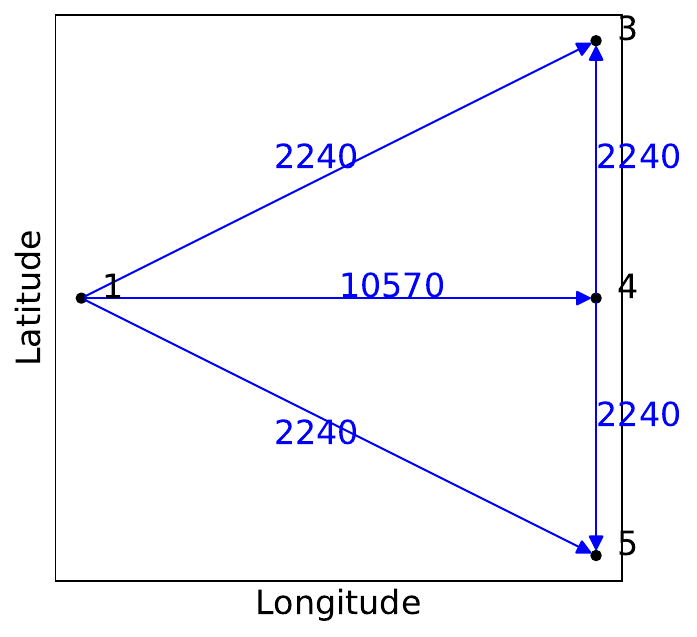}
        \caption{No route overlapping, MNL}
    \end{subfigure}
    \begin{subfigure}[b]{0.28\textwidth}
    \includegraphics[width=\linewidth]{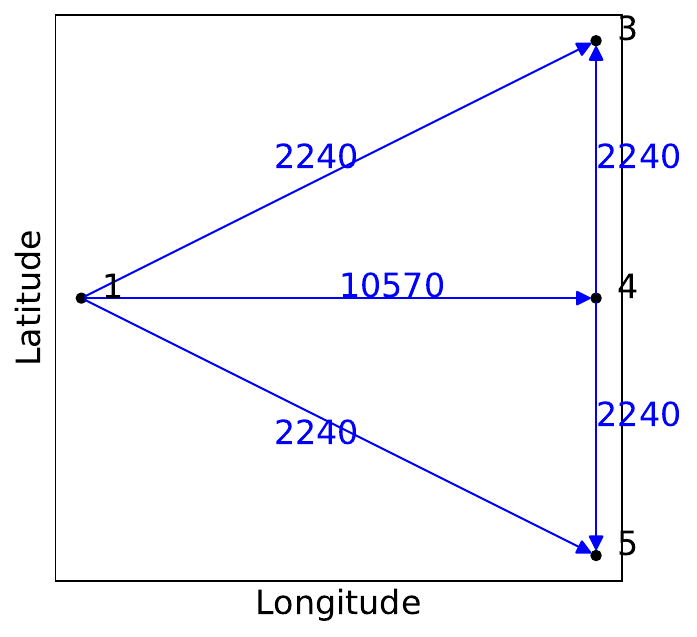}
        \caption{No route overlapping, PSL}
    \end{subfigure}
    \centering
    \captionsetup[subfigure]{skip=0pt}
    % \caption{(continued)}
\end{figure}
\begin{figure}[H]\ContinuedFloat
    \centering
    \captionsetup[subfigure]{skip=0pt}
    \begin{subfigure}[b]{0.28\textwidth}
        \includegraphics[width=\linewidth]{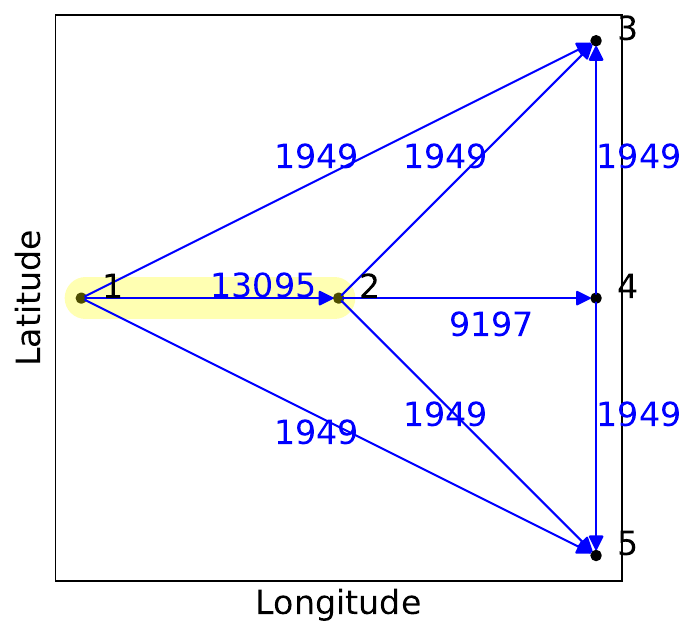}
        \caption{Moderately overlapped routes, MNL}
    \end{subfigure}
    \begin{subfigure}[b]{0.28\textwidth} 
    \includegraphics[width=\linewidth]{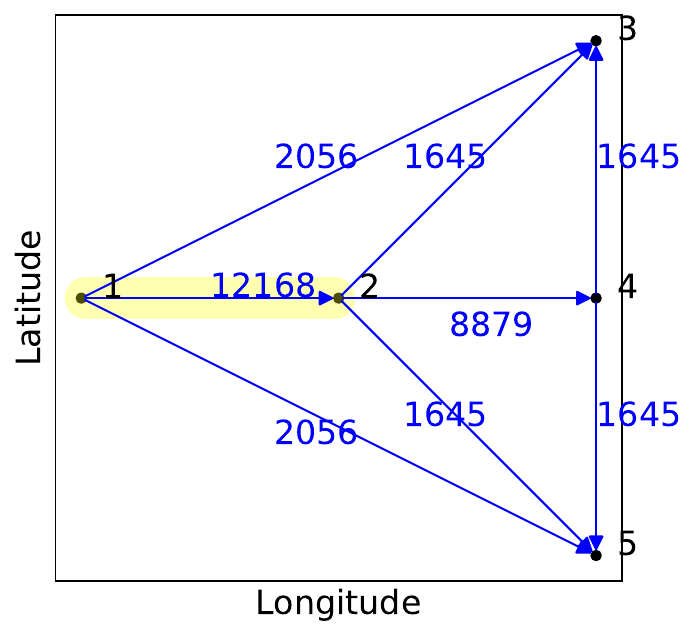}
        \caption{Moderately overlapped routes, PSL}
    \end{subfigure}
    \par\medskip
    \begin{subfigure}[b]{0.28\textwidth}
        \includegraphics[width=\linewidth]{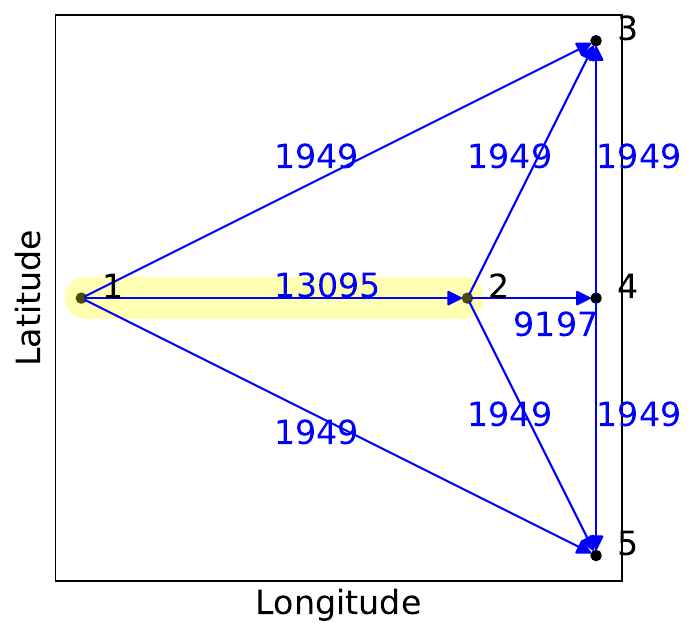}
        \caption{Heavily overlapped routes, MNL}
    \end{subfigure}
    \begin{subfigure}[b]{0.28\textwidth}
        \includegraphics[width=\linewidth]{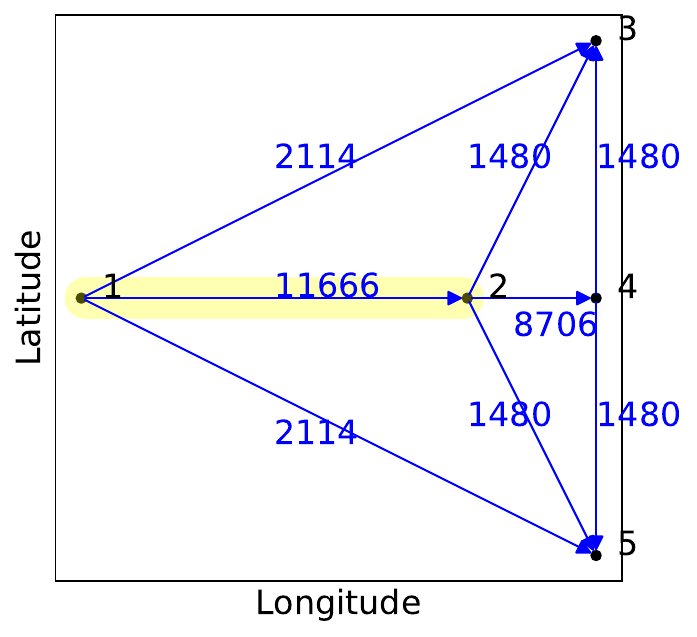}
        \caption{Heavily overlapped routes, PSL}
    \end{subfigure}
    \caption{Traffic counts for automobiles with varying network configurations, each exhibiting different levels of route overlapping. Let $x$ be the length of overlapped arc $(1 \rightarrow 2)$. (a), (b) have $x=0$ and node \#2 disappears. (c), (d) have moderately overlapped routes ($x=2$). (e), (f) have heavily overlapped routes ($x=3$). }
\end{figure}

\subsection{Sioux Falls network}

The Sioux Falls network comprises 24 nodes, 76 links, and 528 OD pairs. The network structure, OD matrix, and the link performance parameters are sourced from \citet{transportnetwork}. 
\YS{The candidate set of routes is generated a priori using the penalty method, with a 5\% penalty applied to the travel times on all links that form the shortest path \citep{bekhor2008effects}. 
In this route set, there are a total of 1,070 routes with an average of 2.02 and a maximum of 5 routes for each OD pair. }

Based on the given OD matrix, we obtain the estimate of $\scalej$ from the dual variable of \cref{FirstStage:ratio1}, which turns out to be 0.1689. As previously mentioned, this value should be less than $\scalem$ and $\scaler$, which are equal to 1, since the correlation at the destination level is stronger than at the lower-level decisions. The value 0.1689 is indeed less than 1, confirming that it is the correct estimation.

\begin{figure}[h!]
    \centering
    \captionsetup[subfigure]{skip=0pt}
    \begin{subfigure}[b]{0.37\textwidth}
        \includegraphics[width=\linewidth]{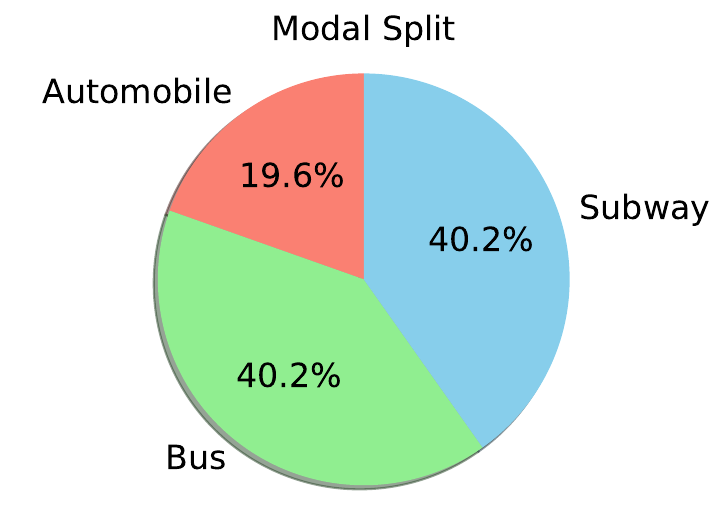}
        \caption{Dissimilarity of public transit ($\tau$) = 1, \\ automobile competitiveness ($\rho$) = 0.9}
    \end{subfigure}
    \begin{subfigure}[b]{0.37\textwidth}
        \includegraphics[width=\linewidth]{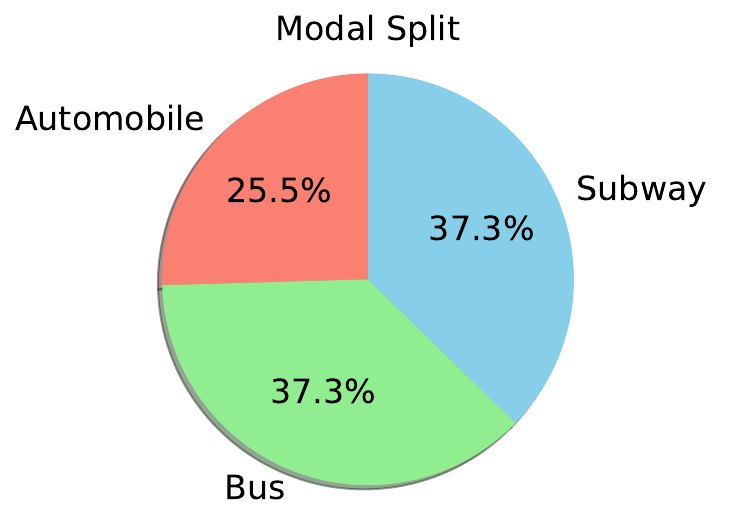}
        \caption{Dissimilarity of public transit ($\tau$) = 0.5, \\ automobile competitiveness ($\rho$) = 0.9}
    \end{subfigure}
    \begin{subfigure}[b]{0.4\textwidth}
        \includegraphics[width=\linewidth]{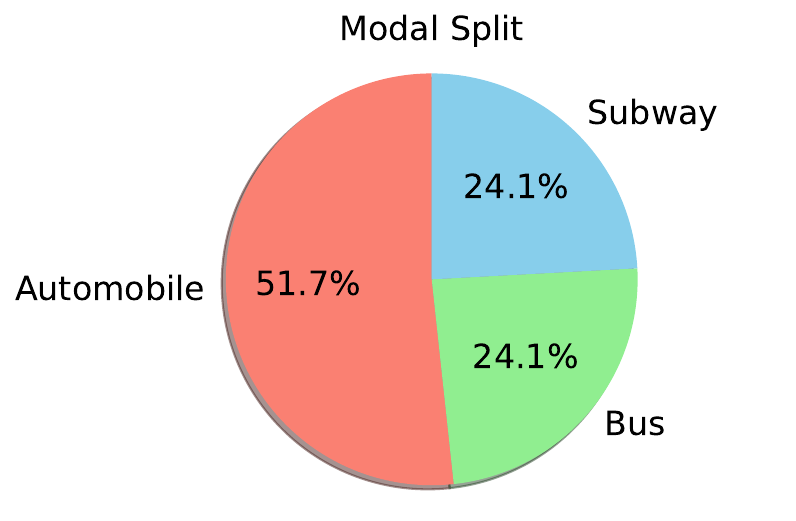}
        \caption{Dissimilarity of public transit ($\tau$) = 1, \\ automobile competitiveness ($\rho$) = 1.1}
    \end{subfigure}
    \begin{subfigure}[b]{0.4\textwidth}
        \includegraphics[width=\linewidth]{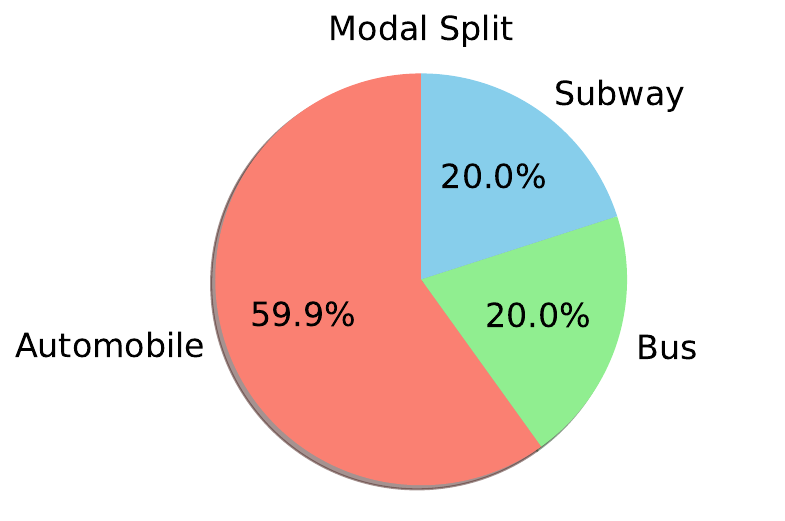}
        \caption{Dissimilarity of public transit ($\tau$) = 0.5, \\ automobile competitiveness ($\rho$) = 1.1}
    \end{subfigure}
\caption{Modal splits for different nest parameters for public transit ($\tau$) and the competitiveness of automobiles ($\rho$)}
\label{fig:sensitivity}
\end{figure}

We conduct sensitivity analysis to demonstrate the impact of employing a NL model over a simpler MNL. 
For our initial modal split sensitivity analysis, we use basic travel time assumptions by simply multiplying the constant $\rho$ with the free flow travel time of automobiles to calculate the travel time for public transit. 
In subsequent analyses, we will incorporate more realistic travel times.
As the travel time of public transit increases, automobiles become more attractive. Thus, we interpret this parameter as the competitiveness of automobiles. 
Another key parameter is $\tau$ (simplified notation for $\tau_{\text{public transit}}$), which represents the dissimilarity among modes within the public transit nest. As $\tau$ decreases, indicating reduced dissimilarity or increased similarity, the modes within the nest become less attractive collectively. This occurs because the nest is perceived as offering less variety.

The results of the sensitivity analysis are presented in \cref{fig:sensitivity}.
It is not surprising that the subway and the bus show exactly the same market share because we used the exact same mode-specific attribute and fixed lines for them in this sensitivity analysis.
As the dissimilarity parameter for public transit ($\tau$) decreases, the competitiveness of the public transit nest reduces, consequently reducing its modal split. When $\tau$ drops from 1 to 0.5, the modal split for bus and subway decreases from 40.2\% to 37.3\% (figure a to b) and from 24.1\% to 20.0\% (figure c to d), respectively. Conversely, when the competitiveness of automobiles increases, its modal split also rises. An increase in automobile competitiveness from 0.9 to 1.1 results in an increase of the modal split for automobiles from 19.6\% to 51.7\% (figure a to c) and from 25.5\% to 59.9\% (figure b to d), respectively. 
It's worth noting that setting $\tau = 1$ effectively reduces our model to a simple MNL. 
Our findings reveal that different values of $\tau$ result in varying estimates for the modal split, implying that an inaccurate $\tau$ could lead to over/under-estimations of the modal split.
This underscores the importance of either estimating (in the first stage model) or utilizing (in the second stage model) a realistic dissimilarity parameter.

Lastly, we conducted a case study to demonstrate the capabilities of our model.
To begin with, we would like to clarify that this study is not intended to draw conclusions or implications that can be directly applied to real-world scenarios. 
Our case study includes three modes of transportation: automobiles, buses, and subways.
The maps for road, bus, and subway networks are shown in \cref{fig:traffic}. 
The benchmark instance for the Sioux Falls network provides the road network layout, free-flow travel times, road capacities, and BPR parameters to account for road congestion. 
These free-flow travel times serve as reference points for other modes. We assume 60 km/hr free-flow speed for automobiles. Subway lines are presumed to be twice as fast as automobiles in free-flow conditions where connections exist. For areas without subway access, we assume a walking speed of 5 km/hr. Passengers choose the shortest path, either by transit, walking, or a combination of both. Buses are assumed to connect all links at an average speed of 15 km/hr, four times slower than free-flow car travel, accounting for access, waiting, and transfer times. We assume passengers take the shortest bus route and are unaffected by congestion due to express lanes. Both bus and subway fares are set at \$2.5. Private car costs, including fuel and mileage, are estimated at \$0.50 per km. We use a value of time of \$20 per hour to convert time units to monetary values. 

After solving our optimization problem, we can extract the complete information on travel demand patterns from the primal solutions, as specific as the arc-level traffic counts in \cref{fig:traffic}.
Annotations in the figure represent the traffic counts at each arc. 
We distinguish each direction of the bidirectional edges by blue and green colors for clarity. For the road network, the width of each arc indicates the congestion level of each road. The level of congestion is calculated based on the volume-to-capacity (V/C) ratio. 
This visualization demonstrates how our framework can be utilized as a powerful tool for analyzing traffic patterns, demand, and congestion at a high level of detail.
\begin{figure}[H]
    \centering
    \includegraphics[width=0.85\textwidth]{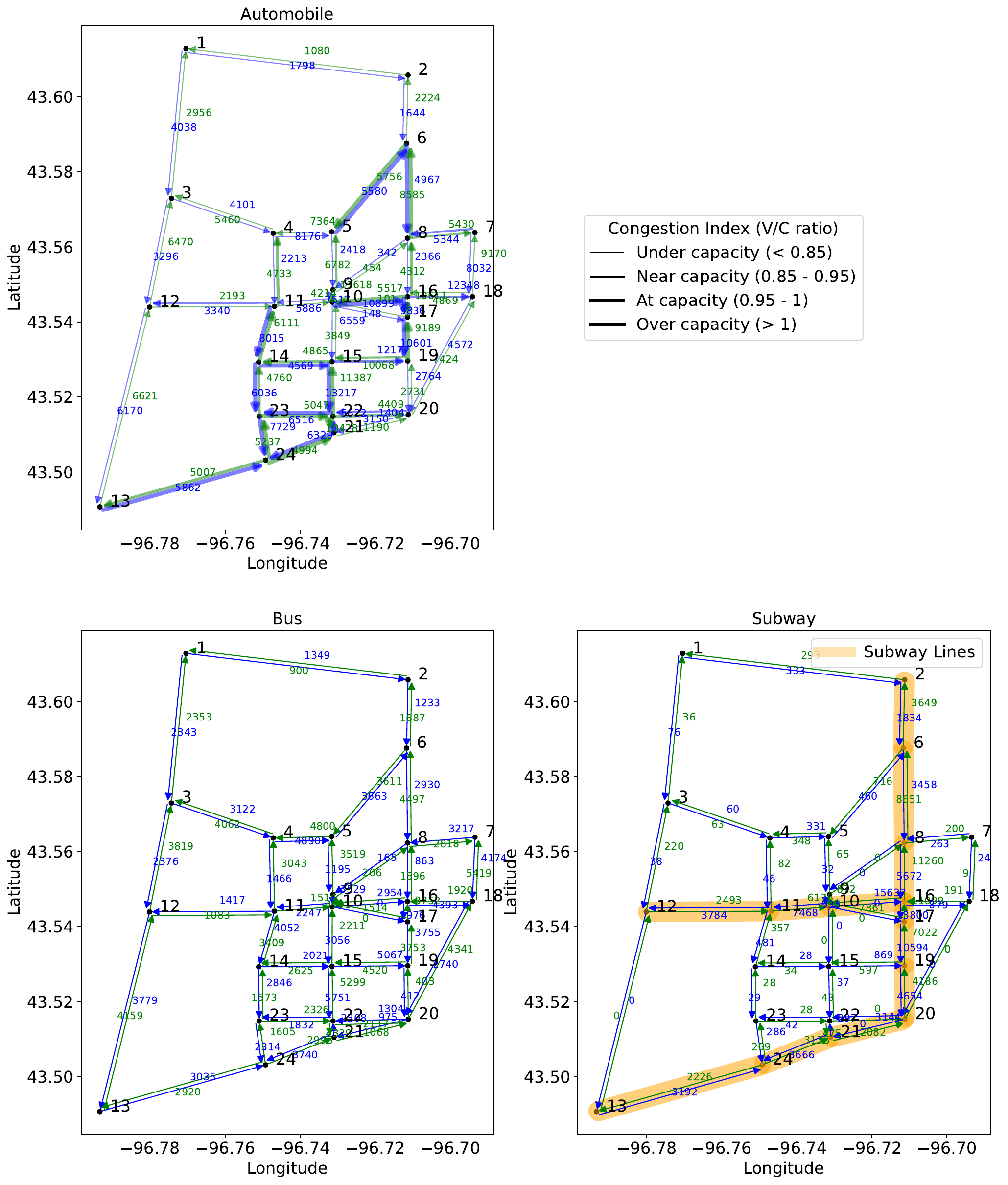}
    \caption{Traffic counts for road, bus, and subway networks. Green and blue colors are used to differentiate bidirectional edges. Annotations represent the traffic counts at each arc. For the road network, the width of each arc indicates the congestion level of each road. }
    \label{fig:traffic}
\end{figure}

\begin{figure}[H]
    \centering
    \begin{subfigure}[b]{0.59\textwidth}
        \includegraphics[width=\textwidth]{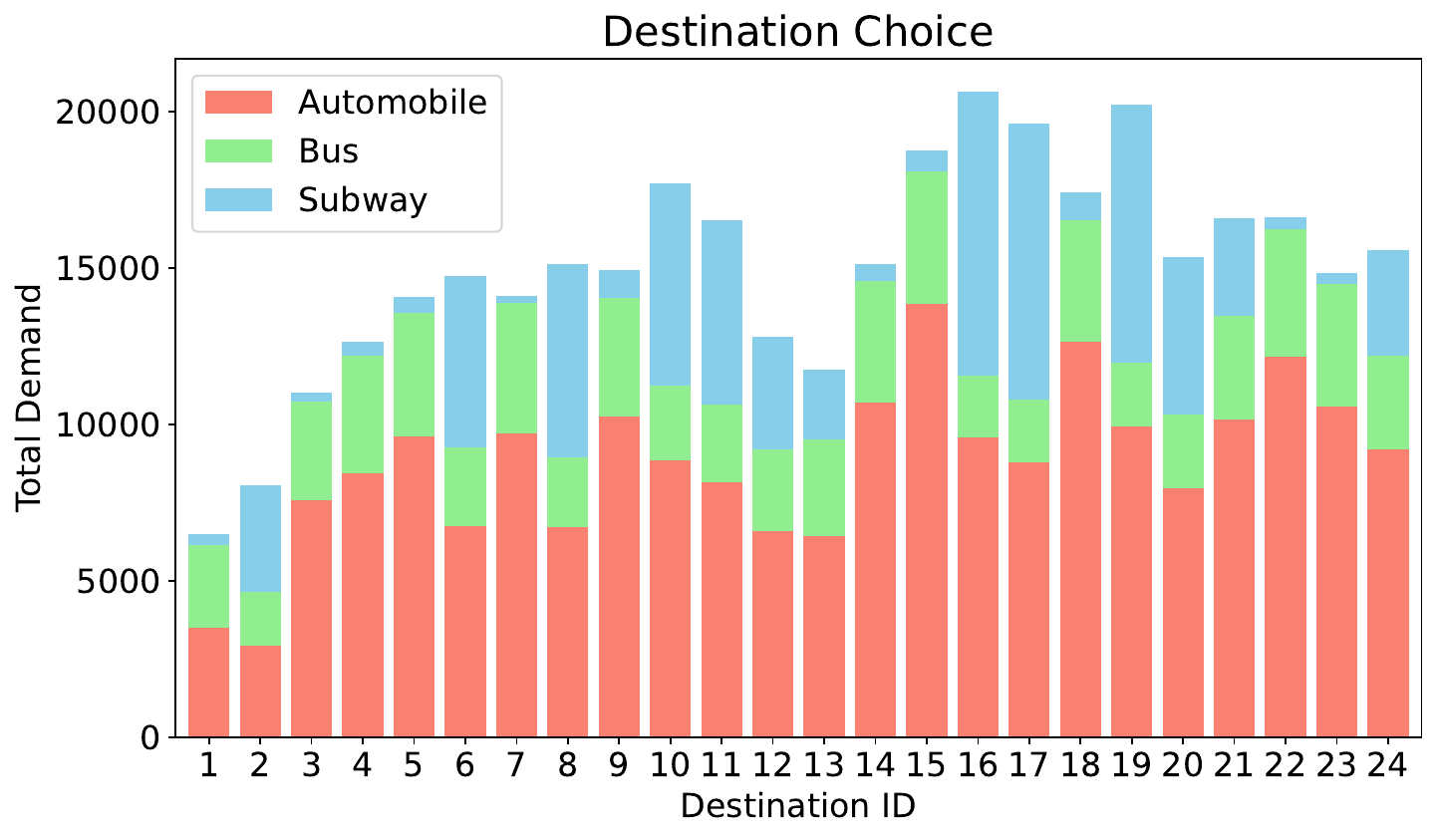}
        \caption{Destination choice}
    \end{subfigure}
    \hfill
    \begin{subfigure}[b]{0.39\textwidth}
        \includegraphics[width=\textwidth]{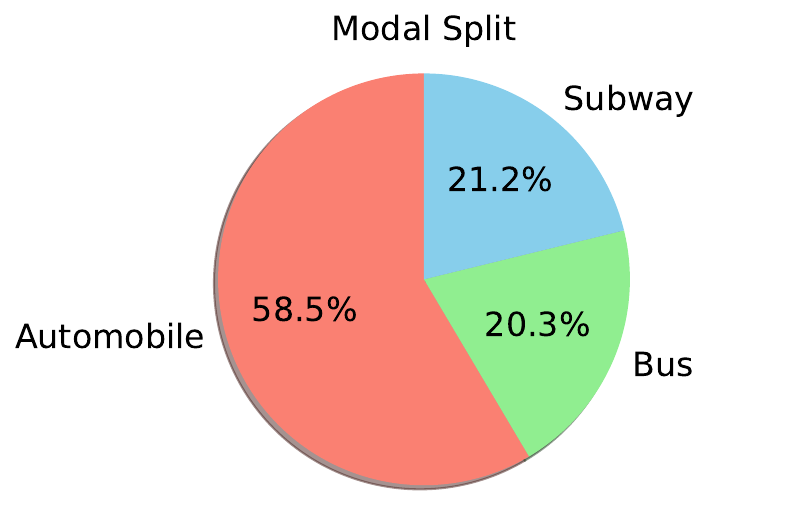}
        \caption{Mode choice}
    \end{subfigure}
    \caption{Aggregated statistics of the network equilibrium}
    \label{figure:visual}
\end{figure}

Now, we visualize the destination choice for each mode and the aggregated modal split.
\Cref{figure:visual} (a) illustrates the total travel demand to each destination. Destinations \#1, \#2, \#12, and \#13 show low demand. These nodes are located on the left boundary of the map, which shows lower traffic. In contrast, destinations \#15, \#16, \#17, \#18, and \#19 show high demand, which makes sense as we observed lots of traffic and congestion in this highly dense area.
Additionally, the modal split at each destination node clearly illustrates a pattern that the nodes directly connected to subway lines (\#2, \#6, \#8, \#10, \#11, \#12, \#13, \#16, \#17, \#19, \#20, \#21, \#24) exhibit higher subway usage rates.
\Cref{figure:visual} (b) shows the modal split for each mode: automobiles are 58.5\%, buses are 20.3\%, and subways are 21.2\%. 
It is important to note that those outputs are subject to change based on the input parameters such as the value of time and perceived fuel and mileage costs. In practical applications, conducting sensitivity analyses for those parameters would provide valuable insights to support decision-making processes.

\YS{\subsection{Various transportation benchmark networks}

Our analysis encompasses several transportation network benchmarks, including Sioux Falls, Eastern Massachusetts (EMA), Berlin Friedrichshain, Berlin Mitte, Anaheim, Barcelona, Winnipeg, and Chicago Sketch. 
\Cref{tab:size} presents the dimensions of each network, detailing the number of nodes, links, OD pairs, and routes.

\YS{
\begin{table}[h!]
\centering
\caption{Size of Benchmark Networks}
\begin{tabular}{|l|r|r|r|r|}
\hline
\textbf{}                      & \multicolumn{1}{l|}{\textbf{\# of nodes}} & \multicolumn{1}{l|}{\textbf{\# of links}} & \multicolumn{1}{l|}{\textbf{\# of ODs}$^a$} & \multicolumn{1}{l|}{\textbf{\# of routes}$^b$} \\ \hline
\textbf{SiouxFalls}            & 24                                        & 76                                        & 528                                     & 644                                        \\ \hline
\textbf{Eastern Massachusetts}                   & 74                                        & 258                                       & 1,113                                   & 1,797                                      \\ \hline
\textbf{Berlin Friedrichshain} & 224                                       & 523                                       & 506                                     & 656                                        \\ \hline
\textbf{Berlin Mitte}          & 398                                       & 871                                       & 1,260                                   & 1,706                                      \\ \hline
\textbf{Anaheim}               & 416                                       & 914                                       & 1,406                                   & 2,700                                      \\ \hline
\textbf{Barcelona}             & 1,020                                     & 2,522                                     & 7,922                                   & 16,570                                     \\ \hline
\textbf{Winnipeg}              & 1,057                                     & 2,535                                     & 4,345                                   & 8,000                                      \\ \hline
\textbf{Chicago Sketch}        & 933                                       & 2,950                                     & 93,513                                  & 225,507                                    \\ \hline
\end{tabular}
\label{tab:size}\\
\footnotesize{$^a$ While the theoretical number of OD pairs equals the square of the number of nodes, we exploit the sparsity of the OD matrix and report only the OD pairs with strictly positive demand. To further utilize the sparsity of the OD matrix, the preprocessing stage filters out demands less than 20.\\
$^b$ The total number of routes is reported. Up to three candidate routes are generated per OD pair. }
\end{table}

}

We successfully solved our model for all networks within a couple of minutes, including the largest Chicago Sketch network with 93,513 OD pairs and 225,507 routes. This provides strong numerical evidence supporting the computational efficiency of solving large-scale combined models using a conic solver.
Because the benchmark datasets do not provide independent observed link counts, the values of $\hat{\lambda}$ reported here are illustrative rather than empirically calibrated. In a real-data application, $\hat{\lambda}$ would instead be selected through the one-dimensional link-count calibration in \cref{metric}. For the benchmark experiments, we therefore report representative values of $\hat{\lambda}$ that yield numerically stable solutions and reasonable scaling between the entropy and Beckmann components.
Through iterative adjustment of the scale parameter, we frequently observe that improper choices of $\scalepar$ can lead to numerical instability, often indicated by a `Slow Progress' flag from the solver and resulting in inaccurate estimates of $\scalej$ that fall outside its theoretical bounds of 0 and 1.
Given the limited availability of data, we present a representative example of $\scalepar$ that balances the Beckmann and entropy components to ensure numerical stability.
Calibrating $\scalepar$ usually involves iterative adjustments to minimize discrepancies between predicted and observed choices \citep[see][\S 11.3.2]{de2011modelling}.

\YS{

\begin{table}[h!]
\caption{Computation Time and Calibration Results for Various Networks}
\centering 
\begin{tabular}{|l|r|r|r|r|r|}
\hline
\multicolumn{1}{|c|}{\textbf{}} & \multicolumn{1}{c|}{\textbf{\begin{tabular}[c]{@{}c@{}}Time to \\ build model \\ (sec)\end{tabular}}} & \multicolumn{1}{c|}{\textbf{\begin{tabular}[c]{@{}c@{}}Solution \\ time \\ (sec)\end{tabular}}} & \multicolumn{1}{c|}{\textbf{\begin{tabular}[c]{@{}c@{}}Objective \\ function\end{tabular}}} & \multicolumn{1}{c|}{\textbf{$\scalepar$}}\\ \hline
\textbf{SiouxFalls}             & 11.0                                                                                                  & 0.7                                                                                             & -5852.8                                                                                     & 1000                                 \\ \hline
\textbf{EMA}                    & 14.4                                                                                                  & 1.3                                                                                             & -527.1                                                                                      & 7000                              \\ \hline
\textbf{Berlin Friedrichshain}  & 16.6                                                                                                  & 0.7                                                                                             & -4698.5                                                                                     & 100                                 \\ \hline
\textbf{Berlin Mitte}           & 18.0                                                                                                  & 0.9                                                                                             & -2088.2                                                                                     & 200                                 \\ \hline
\textbf{Anaheim}                & 55.4                                                                                                  & 3.7                                                                                             & -15436.4                                                                                    & 0.5                              \\ \hline
\textbf{Barcelona}              & 583.9                                                                                                 & 18.8                                                                                            & -6431.6                                                                                     & 3000                               \\ \hline
\textbf{Winnipeg}               & 294.5                                                                                                 & 16.2                                                                                            & -2254.3                                                                                     & 500                               \\ \hline
\textbf{Chicago Sketch}         & 1799.6                                                                                                & 123.1                                                                                           & -2095.0                                                                                     & 5000                            \\ \hline
\end{tabular}
\end{table}
}

\HL{The purpose of the numerical experiments is to demonstrate the behavioral properties and computational tractability of the proposed convex formulation, rather than to empirically validate behavioral parameter recovery. A meaningful validation of the estimated parameters requires multi-source real-world observations with independently observed travel patterns and link counts, which is beyond the scope of this study and is left for future empirical work.}

\subsection{Scalability considerations for real-world networks.}

To ensure scalability to large-scale, real-world transportation networks, potentially involving hundreds of thousands of links, OD pairs, and routes, we provide several technical considerations. 
As the number of probability-related variables increases, individual probabilities tend to approach zero. This, in turn, can result in excessively large parameters, potentially leading to numerical instability in computation. 
To mitigate this issue, practitioners may adopt a combination of the following strategies, tuning them to the specifics of their application:

\textbf{Zone aggregation and network abstraction.} Nodes typically represent aggregated spatial units. Appropriately aggregating these zones can reduce computational complexity and enhance interpretability. For instance, in modeling nationwide travel demand, using states or regions as zones may be more suitable than finer spatial units such as cities or traffic analysis zones.

\textbf{Exploiting OD Matrix sparsity.} Since decision variables are indexed by OD pairs, it is beneficial to exploit the sparsity of the OD matrix. Although the theoretical number of OD pairs scales quadratically with the number of nodes, many pairs may be unused or negligible in practice (e.g., trips between distant or disconnected regions). 
Excluding such pairs helps reduce dimensionality and enables the concentration of computational resources on OD pairs that are most critical to the analysis.

\textbf{Candidate route pruning.} Limiting the number of candidate routes per OD pair can significantly reduce problem size. While a larger route set allows for finer representation of route choice behavior, its marginal benefit diminishes due to cognitive constraints on travelers, who typically consider only a limited number of alternatives. Therefore, the set of candidate routes should be carefully curated and validated using empirical data.
}

\section{Conclusion}
\label{sec:conclusions}
\subsection{Summary}
The travel demand forecasting model plays a critical role in assessing large-scale infrastructure projects. 
The sequential four-step travel demand modeling process remains widely used despite its inherent limitations, including the absence of a unifying framework and discrepancies between input and output values in each step. Although combined models have been developed, they have not been practically applicable due to a lack of satisfaction in both behavioral richness and computational efficiency, which are desirable in the real world. 
To bridge this gap, we proposed a convex programming approach equivalent to the hierarchical extended logit model. 
The most significant advantage of our method lies in its ability to be modeled as convex programming, ensuring solution existence and tractability.

Our model is built upon the utility maximization framework, leveraging well-established findings. Specifically, an individual choice of destination, mode, and route can be expressed as random utility, and then the trip generation/distribution, modal split, and traffic assignment can be viewed as maximizing the aggregated utility. Building upon that, we can extend the well-known equivalence of the multinomial logit model and entropy maximization problem to the nested logit and path size logit model. This extension allows our model to exhibit greater behavioral richness compared to previous convex programming approaches. Each step is expressed using appropriate models, such as the multinomial logit model (which is equivalent to the gravity model) for destination choice, the nested logit model for mode choice, and the path-size logit model for route choice.

The combined model demonstrates its advantages during the two stages of parameter estimation and future prediction. In the parameter estimation stage, it provides an alternative estimation architecture in which empirical observations are incorporated through moment and conditional-entropy constraints. This approach incorporates empirical observations through moment and conditional-entropy constraints, while the associated dual variables recover the corresponding behavioral parameters. Unlike traditional models with feedback loops, our convex programming formulation jointly represents all travel levels within a unified optimization framework, thereby maintaining consistency across destination, mode, and route choices. Moreover, the framework provides a systematic way to reconcile theoretical network equilibrium with empirical observations, including observed link counts, within a unified forecasting framework.

\subsection{Complementary relationship between our framework and activity-based models}

In this section, we discuss the complementary relationship between our suggested framework and the popular activity-based model. The activity-based model can provide comprehensive information to validate our framework, while our framework's output can serve as input for the activity-based model to simulate more complex patterns. Activity-based models simulate whole-day daily activities of individuals or households to understand travel behavior and demand. They focus on capturing the sequence and timing of activities and resulting travel, offering a holistic view of travel behavior. These models excel in defining heterogeneous agents with diverse behaviors and decision-making processes.

However, the complex nature of activity-based models presents calibration and validation challenges due to their heavy reliance on input data quality and parameters. Also, the activity-based model may face computational challenges in the feedback loop accounting for road congestion effects on activities. Our framework serves as a trip-based model, providing parameter estimates from a compact aggregate dataset, though at the expense of detailed activity information. Our framework's first stage problem provides estimated input parameters and an  matrix, which can be used as input for activity-based models capable of modeling more intricate trip activities. 

Conversely, activity-based model outputs can serve as a dataset reflecting real-world travel patterns to validate our model. 
The rich information from the activity-based model will help us better characterize our framework through sensitivity analysis to determine the number of observations needed to fully recover actual patterns.
Using activity-based models to validate our framework presents an interesting avenue for future research. In conclusion, despite their individual limitations, both the activity-based model and our suggested framework can be used collaboratively to work towards the common goal of estimating and forecasting real-world travel patterns.

\subsection{Future research}

\YS{
Validation using multi-source real-world data remains as future investigations. 
During this process, to enhance the robustness of the optimized distribution, we may relax the assumption that the error terms follow the Gumbel distribution. 
Rather than specifying a particular distributional form, we can construct an ambiguity set centered around the empirical distribution. 
This reformulation establishes a natural connection to distributionally robust optimization.
}

In future research, exploring the application of this model to various transportation networks in different cities would provide highly valuable insights. Additionally, conducting a comparative analysis between the performance of the convex programming model and traditional four-step sequential approaches, as well as other advanced travel demand forecasting techniques like activity-based modeling or machine learning approaches, could offer valuable assessments of their strengths and limitations.

For a more comprehensive understanding of travelers' behavior, there are promising extensions to consider. Firstly, accommodating multi-modal trip chains allows for a more comprehensive capture of travelers' behavior. Secondly, investigating dynamic travel demand modeling, where travel patterns evolve over time would better reflect real-world dynamics. To address the challenge of non-separable link cost functions, where the link travel time of each mode depends on the flows of other modes on that link, further research could explore effective ways to handle these complexities, thereby expanding the model's applicability.

The extended model's capability to capture travelers' responses to different policy interventions, such as tolling and congestion pricing, holds significant potential for achieving sustainable urban mobility goals. Additionally, the model can be utilized to evaluate the investment in emerging technologies, such as autonomous vehicles and electric vehicles. By exploring these extensions and applications, the proposed convex programming model can evolve into a powerful and versatile tool, aiding transportation planning and policy-making in diverse urban contexts.

% directly put it into the system
% \section*{CRediT authorship contribution statement}
% \textbf{Youngseo Kim}: Conceptualization, Formal analysis, Investigation, Methodology, Software, Writing – original draft.
% \textbf{Gioele Zardini}: Investigation, Validation, Writing – review and editing.
% \textbf{Samitha Samaranayake}: Funding acquisition, Validation, Supervision, Writing – review and editing.
% \textbf{Soroosh Shafiee}: Formal analysis, Investigation, Methodology, Writing – original draft.

\section*{Acknowledgements}

\HL{This research is partially sponsored by the National Science Foundation under grants CIS-2144127, the Department of Energy under Award Numbers DE-EE0010609, the 2023 INFORMS TSL Cross-Regional Doctoral Grant, the UC Institute of Transportation Studies under Award No. 2026-16, and the U.S. Department of Transportation through the Pacific Southwest Region University Transportation Center.}

We extend our sincere gratitude to Dr. Damon Wischik for providing the initial inspiration for the entropy maximization framework. We also thank to Dr. Sedong Moon and Eui-Jin Kim for the insightful conversation regarding the current four-step sequential modeling approach. We also thank the two anonymous reviewers of Transportation Research Part B for their thoughtful and constructive feedback throughout the review process, which significantly improved the clarity and presentation of the manuscript.

\bibliographystyle{plainnat}
\bibliography{references}

\newpage
\begin{appendices}

\section{Proofs}
\subsection{\texorpdfstring{Proof of \cref{theorem:convex}}{Proof of Theorem 1}}
\label{appendix:hier_extend}

\renewcommand{\theequation}{A.\arabic{equation}}
\setcounter{equation}{0}

The formulation follows the standard form of a convex optimization problem, involving the maximization of a concave objective function subject to convex inequality constraints and linear equality constraints. First, the objective function is concave because the conditional entropy functions are concave \citep{polyanskiy2024information}, and the negative of the Beckmann equation is concave under Assumption~1 because the Beckmann equation itself is convex \citep{beckmann1956studies}. Second, the feasible sets defined by the conditional-entropy constraints are convex because they are superlevel sets of concave functions. Third, all equality constraints are linear. Consequently, the formulation is a convex optimization problem, and the \gls{a:kkt} conditions characterize the optimal solutions under standard constraint qualifications. By solving the \gls{a:kkt} conditions, we recover the hierarchical extended logit model.

Let $\lambda_i, \mu_{ij}, \kappa_{ij\nest}, \nu_{ijm}$ be the dual variables correspond to the constraints in \Pset. Write the Lagrangian function. 
\begin{align*}
    & \lag (\jprobj, \jprobm, \jprobnest, \jprobr, \scalej, \scalem, \nestpar, \tripattrpar, \tripmattrpar, \duali, \dualj, \dualM, \dualm)  \\
    & = - \frac{\scalepar}{\obstotal} \beckmann \\
    & - \frac{1}{\scalej} [\sum_{i \in \iset, j \in \jset}  \optjprobj \ln (\frac{\optjprobj}{\obsprobi}) ] + (1 - \frac{1}{\scalej}) H_{\text{MNL}}(\obscprobvecj) \\
    & - \frac{1}{\scalem} \left[ \sum_{i \in \iset, j \in \jset} \sum_{\nest \in \nestset} \optjprobnest \ln (\frac{\optjprobnest}{\optjprobj}) \right. \left. \quad \quad 
    + \sum_{\nest \in \nestset} \nestpar [\sum_{i \in \iset, j \in \jset} \sum_{m \in \nest}  \optjprobm \ln (\frac{\optjprobm}{\optjprobnest}) ] \right] \\ 
    & + (1 - \frac{1}{\scalem}) H_{\text{NL}_1}(\obscprobvecnest) 
    + (1 - \frac{1}{\scalem}) \sum_{\nest \in \nestset} H_{\text{NL}_{\nest}}(\obscprobvecm) \\
    & - [\sum_{i \in \iset, j \in \jset} \sum_{m \in \mset, r \in \rset} \optjprobr \ln (\frac{\optjprobr}{\optjprobm} \cdot \frac{1}{\pathsizepar}) \\ 
    & + \sum_{k \in \kset} \tripattrpar [\sum_{i \in \iset, j \in \jset} (\optjprobj \tripattr - \obsjprobj \tripattr)] \\
    & + \sum_{q \in \qset} \tripmattrpar [\sum_{i \in \iset, j \in \jset} \sum_{m \in \mset} (\optjprobm  \tripmattr - \obsjprobm  \tripmattr )] \\
    & + \sum_{i \in \iset} \optduali [ \sum_{j \in \jset} \optjprobj ] \\
    & + \sum_{i \in \iset, j \in \jset} \optdualj [\optjprobj - \sum_{\nest \in \nestset} \optjprobnest] \\
    &  + \sum_{i \in \iset, j \in \jset} \sum_{\nest \in \nestset} \optdualM [\optjprobnest - \sum_{m \in \nest} \optjprobm]\\
    & + \sum_{i \in \iset, j \in \jset} \sum_{m \in \mset} \optdualm [\optjprobm - \sum_{r \in \rsetfull} \optjprobr ]  
\end{align*}

At optimality, all the equality constraints in \cref{FirstStage:prob1}-\cref{FirstStage:prob4} for joint probability are satisfied. 
From the optimal joint probability \optjprobj, \optjprobnest, \optjprobm, \optjprobr, we construct the conditional probability \optcprobj, \optcprobnest, \optcprobmnest, \optcprobr\ from definitions.

\begin{align*}
    \optcprobj & = \frac{\optjprobj}{\sum_{j \in \jset}\optjprobj} = \frac{\optjprobj}{\obsprobi}, \forall i\\
    \optcprobnest & = \frac{\optjprobnest}{\sum_{\nest \in \nestset}\optjprobnest} = \frac{\optjprobnest}{\optjprobj}, \forall i,j\\
    \optcprobmnest & = \frac{\optjprobm}{\sum_{m \in \nest} \optjprobm} = \frac{\optjprobm}{\optjprobnest}, \forall i,j,\nest\\
    \optcprobr & = \frac{\optjprobr}{\sum_{r \in \rsetfull }\optjprobr} = \frac{\optjprobr}{\optjprobm}, \forall i,j,m
\end{align*}

Using the chain rule for differentiation, 
\begin{equation*}
    \frac{\partial \lag}{\partial \optjprobj} = 0 \Rightarrow \frac{\partial \lag}{\partial \optcprobj} \cdot \frac{\partial \optcprobj}{\partial \optjprobj} = 0.
\end{equation*}

Since $\frac{\partial \optcprobj}{\partial \optjprobj} \neq 0$, we conclude that $\frac{\partial \lag}{\partial \optcprobj} = 0$. 
Similarly, we obtain $\frac{\partial \lag}{\partial \optcprobnest} = 0$, $\frac{\partial \lag}{\partial \optcprobmnest} = 0$, and $\frac{\partial \lag}{\partial \optcprobr} = 0$.

Now, we can express the Lagrangian function in terms of conditional probabilities and then find the first-order condition by taking the derivative with respect to the conditional probabilities.
For simplicity, let us denote all the constant terms that are independent of the joint probabilities as $C$. This is because we are primarily interested in the first-order derivatives with respect to the conditional probabilities, and the constant terms will disappear when calculating the first-order derivatives.
\begin{align*}
    \lag =& - \frac{\scalepar}{\obstotal} \beckmann \\
    & - \frac{1}{\scalej} [\sum_{i \in \iset, j \in \jset} \obsprobi \optcprobj \ln (\optcprobj) ] \\
    & - \frac{1}{\scalem} \left[ \sum_{i \in \iset, j \in \jset} \sum_{\nest \in \nestset} \obsprobi \optcprobj \optcprobnest \ln (\optcprobnest)  \right. \\
    & \left. \quad \quad + \sum_{\nest \in \nestset} \nestpar [\sum_{i \in \iset, j \in \jset} \sum_{m \in \nest} \obsprobi \optcprobj \optcprobnest \optcprobmnest \ln (\optcprobmnest) ] \right] \\
    & - [\sum_{i \in \iset, j \in \jset} \sum_{\nest \in \nestset} \sum_{m \in \nest} \sum_{r \in \rsetfull } \obsprobi \optcprobj \optcprobnest \optcprobmnest \optcprobr \ln ( \optcprobr \cdot \frac{1}{\pathsizepar})] \\ 
    & + \sum_{k \in \kset} \opttripattrpar [\sum_{i \in \iset, j \in \jset} \obsprobi \optcprobj \tripattr] \\
    & + \sum_{q \in \qset} \opttripmattrpar [\sum_{i \in \iset, j \in \jset, m \in \mset} \obsprobi \optcprobj \optcprobnest \optcprobmnest \tripmattr ] \\
    & + \sum_{i \in \iset} \optduali [\obsprobi \sum_{j \in \jset} \optcprobj] \\
    & + \sum_{i \in \iset, j \in \jset} \optdualj [ \obsprobi \optcprobj - \sum_{\nest \in \nestset} \obsprobi \optcprobj \optcprobnest] \\
    &  + \sum_{i \in \iset, j \in \jset} \sum_{\nest \in \nestset} \optdualM [ \obsprobi \optcprobj \optcprobnest -  \sum_{m \in \nest} \obsprobi \optcprobj \optcprobnest \optcprobmnest]\\
    & + \sum_{i \in \iset, j \in \jset} \sum_{m \in \mset} \optdualm [ \obsprobi \optcprobj \optcprobnest \optcprobmnest -  \sum_{r \in \rset} \obsprobi \optcprobj \optcprobnest \optcprobmnest \optcprobr ]  + C
\end{align*}

Let us focus on the first term. Recall that 

\begin{align*}
    \arcflow = \sum_{i \in \iset, j \in \jset} \sum_{r \in \rsetfull} \obstotal \jprobr \ind = \sum_{i \in \iset, j \in \jset} \sum_{r \in \rsetfull} \obstotal \obsprobi \cprobj \cprobnest \cprobm \cprobr \ind .
\end{align*}

The partial derivative with respect to $\optcprobr$ is 

\begin{align*}
    & \frac{\scalepar}{\obstotal} \frac{\partial \left( \beckmann \right)}{\partial \optcprobr} \\
    &= \scalepar \sum_{a \in r} \latency (\arcflow) \cdot \frac{\partial \arcflow}{\partial \optcprobr} 
    = \scalepar \obsprobi \optcprobj \optcprobnest \optcprobmnest \sum_{a \in \rset} \latency (\arcflow) \ind \nonumber 
     = \scalepar \obsprobi \optcprobj \optcprobnest \optcprobmnest \cdot \rcost,
\end{align*}

the partial derivative with respect to $\optcprobmnest$ is

\begin{align*}
    & \frac{\scalepar}{\obstotal} \frac{\partial \left( \beckmann \right) }{\partial \optcprobmnest} \\
    & = \scalepar \sum_{a \in r} \latency(\arcflow) \cdot \frac{\partial \arcflow}{\partial \optcprobmnest} 
    = \scalepar \obsprobi \optcprobj \optcprobnest \sum_{r \in \rset} \optcprobr \sum_{a \in \rset} \latency(\arcflow) \ind \nonumber \\
    & = \scalepar \obsprobi \optcprobj \optcprobnest \sum_{r \in \rset} \optcprobr \cdot \rcost = \scalepar \obsprobi \optcprobj \optcprobnest \cdot \rcost, 
\end{align*}

the partial derivative with respect to $\optcprobnest$ is
\begin{align*}
    & \frac{\scalepar}{\obstotal} \frac{\partial \left( \beckmann \right) }{\partial \optcprobnest} \\
    & = \scalepar \obsprobi \optcprobj \sum_{m \in \mset} \optcprobmnest \sum_{r \in \rset} \optcprobr \cdot \rcost 
    = \scalepar \obsprobi \optcprobj \cdot \rcost, 
\end{align*}

and the partial derivative with respect to $\optcprobj$ is

\begin{align*}
    & \frac{\scalepar}{\obstotal} \frac{\partial \left( \beckmann \right) }{\partial \optcprobj} \\
    & = \scalepar \obsprobi \sum_{\nest \in \nestset} \optcprobnest \sum_{m \in \nest} \optcprobmnest \sum_{r \in \rset} \optcprobr \cdot \rcost
    = \scalepar \obsprobi \cdot \rcost. 
\end{align*}

Obtain the first order condition, i.e., $\frac{\partial \lag}{\partial \optcprobr} = 0$. 
\begin{align*}
    & \frac{\partial \lag}{\partial \optcprobr}  \\
    &= - \scalepar \obsprobi \optcprobj \optcprobnest \optcprobmnest \rcost
    - \obsprobi \optcprobj \optcprobnest \optcprobmnest \left( \ln(\frac{\optcprobr}{\pathsizepar}) + 1 \right) 
     - \obsprobi \optcprobj \optcprobnest \optcprobmnest \optdualm = 0 \\
    & \Rightarrow  - \scalepar \rcost - \left[ \ln(\optcprobr) - \ln(\pathsizepar) + 1 \right] - \optdualm = 0 \\
    & \Rightarrow \ln(\optcprobr) = \left( - \scalepar \rcost - \optdualm \right) + \ln(\pathsizepar) - 1\\
    & \Rightarrow \optcprobr = \frac{ \exp(- \scalepar \rcost + \ln(\pathsizepar))}{\exp(1 + \optdualm)} 
\end{align*}

Using that $\sum_{r \in \rset} \optcprobr = 1$, we obtain 
\begin{align*}  
    \optcprobr = \frac{ \exp(- \scalepar \rcost + \ln(\pathsizepar) )}{\sum_{r \in \rset} \exp(- \scalepar \rcost + \ln(\pathsizepar) )}, 
\end{align*}

which is
\begin{align*}
    \frac{ \exp( \cutilr )}{\sum_{r' \in \rset} \exp( V_{r'|ijm} ) }
\end{align*}
by the definition of fixed utility $\cutilr = - \scalepar \rcost + \ln(\pathsizepar) $ in \cref{eqn:utilroute}. 

Using the definition of $\satistripm = \ln \left( \textstyle \sum_{r \in \rsetfull} \exp({ \cutilr}) \right)$, the above equation can be written as below for the later convenience. 
\begin{align}\label{proof_HL:p_rijm_variant}
    \ln (\optcprobr) = \cutilr - S_{ijm} 
\end{align}

Obtain the first order condition, i.e., $\frac{\partial \lag}{\partial \optcprobmnest} = 0$. 
\begin{align*}
    \frac{\partial \lag}{\partial \optcprobmnest} 
    &=
    - \scalepar \obsprobi \optcprobj \optcprobnest \cdot \rcost 
    -  \left[ \obsprobi \optcprobj \optcprobnest \sum_{r \in \rset} \optcprobr \ln(\frac{\optcprobr}{\pathsizepar}) \right] \\
    & - \frac{\optnestpar}{\optscalem} \left[ \obsprobi \optcprobj \optcprobnest \left(\ln(\optcprobmnest) + 1\right) \right] 
    \\ & + \obsprobi \optcprobj \optcprobnest \sum_{q \in \qset} \opttripmattrpar \tripmattr + \obsprobi \optcprobj \optcprobnest (-\optdualM + \optdualm) = 0 
\end{align*}

Using \cref{proof_HL:p_rijm_variant},
\begin{align*}
    & \Rightarrow  \frac{\optnestpar}{\optscalem} \left(\ln(\optcprobmnest) + 1\right) = \sum_{q \in \qset} \opttripmattrpar \tripmattr + (-\optdualM + \optdualm) + S_{ijm} \\
    & \Rightarrow \ln(\optcprobmnest) = \optscalem \left[\sum_{q \in \qset} \opttripmattrpar \tripmattr + S_{ijm}\right]/ \optnestpar + \frac{\optscalem (-\optdualM + \optdualm)}{\optnestpar} - 1\\
    & \Rightarrow \optcprobmnest = \frac{\exp\left(\optscalem \left[\sum_{q \in \qset} \opttripmattrpar \tripmattr + S_{ijm}\right]/ \optnestpar\right)}{\exp(1-\frac{\optscalem (-\optdualM + \optdualm)}{\optnestpar})}  
\end{align*}

Using $\sum_{m \in \nest} \optcprobmnest = 1$, 
\begin{align*}
    \optcprobmnest = \frac{\exp \left(\optscalem \left[\sum_{q \in \qset} \opttripmattrpar \tripmattr + S_{ijm}\right]/ \optnestpar \right)}
    {\sum_{m' \in \nest} \exp\left(\optscalem \left[\sum_{q \in \qset} \opttripmattrpar X^q_{ijm'} + S_{ijm'}\right]/ \optnestpar\right)}, 
\end{align*}

which is 

\begin{align*}
    \optcprobmnest = \frac{\exp \left(\optscalem \left[ \cutilm + S_{ijm}\right]/ \optnestpar \right)}
    {\sum_{m' \in \nest} \exp\left(\optscalem \left[V_{m'|ij}  + S_{ijm'}\right]/ \optnestpar\right)}, 
\end{align*}

by the definition of the fixed utility $\cutilm = \sum_{q \in \qset} \opttripmattrpar \tripmattr$ in \cref{eqn:utilmode}.

Similarly, from $\frac{\partial \lag}{\partial \optcprobnest} = 0$, we derive
\begin{align*}
    \optcprobnest = \frac{\exp(IV_{\nest})}{\sum_{\nest' \in \msetfull} \exp(IV_{\nest'})},
\end{align*}

where $IV_{\nest} = \ln \left(\sum_{m \in \nest} \exp(\optscalem[ \cutilm + S_{ijm}]/ \optnestpar) \right)$. 

Lastly, from $\frac{\partial \lag}{\partial \optcprobj} = 0$, we obtain 
\begin{align*}
    \optcprobj = \frac{\exp\left(\optscalej (S_{ij} +  \cutilj )\right)}{\sum_{j' \in \jset} \exp\left(\optscalej (S_{ij'} + V_{j'|i} \right)},
\end{align*}

where $\cutilj = \sum_{k \in \kset} \opttripattrpar \tripattr$.

\subsection{\texorpdfstring{Proof of \cref{theorem:expcone}}{Proof of Theorem 2}}
\label{appendix:expo_cone}

The objective function comprises the conditional entropy terms and the integral of the latency function. 
First, the conditional entropy terms can be reformulated as exponential cones using \cref{lemma:entropy}. 
Second, the integral of the latency function under Assumption \ref{assume:latency} remains a power function, which is exponential cone representable according to \cref{lemma:power}.

\begin{lemma}\label{lemma:entropy}
    The hypograph of the entropy function is exponential cone representable \citep[see][\S 3]{aps2024mosek}.
\end{lemma}
\begin{proof} 
    Consider the hypograph of the entropy function, $H(\mathbf{p}) \geq t$ where $H(\mathbf{p}) = \sum_{i=1}^n - p_i \ln (p_i) $. There exists $\mathbf{y} \in R^n$ such that $\sum_{i=1}^n y_i =t$, $y_i \leq -p_i \ln p_i$ that can be rewritten as $y_i \leq p_i \ln (\frac{1}{p_i})$ for all $i$. Thus, $(1, p_i, y_i) \in K_{exp}, \forall i$ and $\sum_{i=1}^n y_i =t$.
\end{proof}

\begin{lemma}\label{lemma:power}
    The epigraph of the convex increasing rational power function, $x^{p/q} \leq t$, is exponential cone representable  \citep[see][\S 5]{aps2024mosek}.
\end{lemma}

\begin{proof}
    The hypergraph of the geometric mean function is $K^n_{gm} = \{ (\mathbf{x},t) \in \mathbb{R}^{n+1}: (x_1 x_2 \dots x_n)^{1/n} \geq t, \mathbf{x} \geq 0 \}$. It is well known that geometric cones are exponential cone representable. Thus, if we show that the epigraph of a convex increasing rational function is a geometric cone, it is exponential cone representable. The inequality $x^{\frac{p}{q}} \leq t$ is equivalent to $x \leq t^{\frac{q}{p}}$, which is a geometric cone $(t\mathbb {1}_p, \mathbb{1}_{p-q}, x) \in K^p_{gm}$, where $\mathbb{1}_n$ denotes the vector of ones of size $n$. 
\end{proof}

These lemmas pave the way for reformulating the entire \textsc{First Stage Problem} as an exponential cone program.
We introduce auxiliary variables $t_{ij}$, $u_{ij\nest}$, $v_{ijm}$, and $w_{ijmr}$ for the exponential cone reformulation and $\auxmain$ and $f_a$ for the power cone reformulation.

First, we show that the conditional entropy terms are exponential cone representable. The exponential cone is defined as 
\begin{align*}
    & y \exp(\frac{x}{y}) \leq z \\
    & \iff (x, y, z) \in \expcone.
\end{align*}

Recall an entropy term in objective function, $H_{\text{MNL}}(\cprobvecj) = \sum_{i \in \iset, j \in \jset} - \jprobj \ln (\frac{\jprobj}{\obsprobi})$. 
\begin{align*}
    & \max H_{\text{MNL}}(\cprobvecj)  = \max \sum_{i \in \iset, j \in \jset} - \jprobj \ln (\frac{\jprobj}{\obsprobi}) 
\end{align*}

By introducing an auxiliary variable $t_{ij}$, we can reformulate the problem as follows.
\begin{align*}
    & \max \sum_{i \in \iset, j \in \jset} t_{ij} \\
    & \text{s.t. }  \frac{\obsprobi}{\jprobj} \geq \exp(\frac{t_{ij}}{\jprobj})
\end{align*}

This implies that the domain lies in an exponential cone, i.e., $(t_{ij}, p_{ij}, \obsprobi) \in \expcone, \forall i, j $. 
Similarly, the other entropy terms can also be expressed using the exponential cone with auxiliary variables $u_{ij\nest}, v_{ijm}, w_{ijmr}$.

Regarding the entropy term $H_{\text{NL}_1}(\cprobvecnest) = \sum_{i \in \iset, j \in \jset} \sum_{\nest \in \nestset} - \jprobnest \ln (\frac{\jprobnest}{\jprobj})$, the objective function 
\begin{align*}
    \max \sum_{i \in \iset, j \in \jset} \sum_{\nest \in \nestset} - \jprobnest \ln (\frac{\jprobnest}{\jprobj})
\end{align*}
can be reformulated as follows. 

\begin{align*}
    \max \sum_{i \in \iset, j \in \jset} \sum_{\nest \in \nestset} u_{ij\nest} \text{  where   } (u_{ij\nest}, \jprobnest, \sum_{\nest \in \nestset}\jprobnest) \in \expcone, \forall i \in \iset, j \in \jset, \nest \in \nestset
\end{align*}

Regarding the entropy term 
$H_{\text{NL}_{\nestofm}}(\cprobvecmnest) = \sum_{i \in \iset, j \in \jset}\sum_{m \in \mset} - \jprobm \ln (\frac{\jprobm}{p_{ij\nest}}), \forall \nest \in \nestset$, the objective function 
\begin{align*}
    & \max \sum_{i \in \iset, j \in \jset}\sum_{m \in \nest} - \jprobm \ln (\frac{\jprobm}{p_{ij\nest}}) 
\end{align*}

can be reformulated as follows. 
\begin{align*}
   \Rightarrow & \max \sum_{i \in \iset, j \in \jset}\sum_{m \in \nest} v_{ijm} \text{  where   } (v_{ijm}, \jprobm, \sum_{m' \in \nest} p_{ijm'}) \in \expcone, \forall i \in \iset, j \in \jset, m \in \mset
\end{align*}

Regarding the entropy term 
$H_{\text{PSL}}(\cprobvecr) = \sum_{i \in \iset, j \in \jset}\sum_{m \in \mset} \sum_{r \in \rsetfull} - \jprobr \ln (\frac{\jprobr}{\sum_{r \in \rsetfull} \jprobr \pathsizepar})$, the objective function 
\begin{align*}
    & \max \sum_{i \in \iset, j \in \jset}\sum_{m \in \mset} \sum_{r \in \rsetfull} - \jprobr \ln (\frac{\jprobr}{\sum_{r \in \rsetfull} \jprobr \pathsizepar}) 
\end{align*}
can be reformulated as follows. 
\begin{align*}
    & \max \sum_{i \in \iset, j \in \jset}\sum_{m \in \mset} \sum_{r \in \rsetfull} w_{ijmr} \text{  where   } (w_{ijmr}, \jprobr, \sum_{r \in \rsetfull} \jprobr \pathsizepar) \in \expcone, \forall i \in \iset, j \in \jset, m \in \mset, r \in \rset
\end{align*}

Second, we show that the integral of the \gls{a:bpr} function is a power cone. 

If the mode is automobiles, the latency function $g_a$ is defined using the standard \gls{a:bpr} function. It is reasonable to assume that subways are unaffected by congestion, and similarly, buses are assumed to be unaffected due to the use of express lanes. In such cases, where $\beta=0$, the latency function provides a constant travel time. Since only automobiles require the standard \gls{a:bpr} function, we will omit the superscript indicating the mode for simplicity. 

\begin{align*} 
    g_a(f_a) = T^0_a \left[1+ \alpha \left(\frac{f_a}{c_a}\right)^{\beta}\right].
\end{align*}

Recall the Beckmann equation in the objective function. 
\begin{align*}
    \sum_{a \in \aset} \int^{f_a}_{0} g_a(w) dw = \sum_{a \in A} \left( T^0_a f_a + \frac{T^0_a \alpha}{(\beta+1) c_a^{\beta}} f^{\beta+1}_a \right)
\end{align*}

For modes that are not affected by congestion, $\alpha = 0$ and $\beta = 0$, resulting in a simple expression $\sum_{a \in A} ( T^0_a f_a)$. 

The power cone $f_a^p$, where $p>1$ ($\beta=4$ in the standard \gls{a:bpr} function, so we consider $f_a^5$), satisfies the inequality $s_a \geq |f_a|^p$, which is equivalent to ${s_a}^{\frac{1}{p}} \geq |f_a|$. Hence, it corresponds to 
\begin{align*}
    {s_a} \geq |f_a|^p \Longleftrightarrow (s_a, 1, f_a) \in \mathcal{P}_3^{\frac{1}{p}, 1-\frac{1}{p}}. 
\end{align*}

The resulting conic reformulation is as follows.

\begin{center}
\fbox{\parbox{0.9\columnwidth}{ {\centering
\vspace{0.5ex}\textsc{Exponential cone reformulation}\\[1ex]}

\begin{align}
    \max  
    &\sum_{i \in \iset, j \in \jset} t_{ij} 
        + \sum_{i \in \iset, j \in \jset, \nest \in \nestset} u_{ij\nest} \nonumber \\ \nonumber
    & + \sum_{\nest \in \nestset}(\sum_{i \in \iset, j \in \jset, m \in \nest} v_{ijm}) 
        + \sum_{i \in \iset, j \in \jset, m \in \mset, r \in \rset} w_{ijmr} \\ 
    &- \frac{1}{\obstotal} \sum_{m \in \mset} \sum_{a \in A} \widehat{T}^{m,0}_a f_a - \frac{1}{\obstotal} \left(\sum_{a \in A}  \frac{\widehat{T}^{\text{car}, 0}_a \alpha}{(\beta+1) c_a^{\beta}} \auxmain \right) \nonumber 
\end{align}

\begin{subequations}
\begin{align}
    \text{s.t.} \nonumber \\  
    & (t_{ij}, p_{ij}, \obsprobi) \in \expcone, \forall i,j  \nonumber\\
    & (u_{ij\nest}, \jprobnest, \sum_{\nest \in \nestset }\jprobnest) \in \expcone, \forall i,j,\nest \nonumber \\
    & (v_{ijm}, \jprobm, \sum_{m' \in \nest | m \in \nest} p_{ijm'}) \in \expcone, \forall i, j, m \nonumber \\
    & (w_{ijmr}, \jprobr, \sum_{r \in \rsetfull} \jprobr) \in \expcone, \forall i, j, m, r \nonumber \\
    & \sum_{ij} t_{ij} = H_{\text{MNL}}(\cprobvecj) \nonumber \\
    & \sum_{ij\nest} u_{ij\nest} = H_{\text{NL}_1}(\cprobvecnest) \nonumber \\
    & \sum_{ijm} v_{ijm} = H_{\text{NL}_{\nest}}
    (\cprobvecmnest), \forall \nest \in \nestset \nonumber \\
    & (s_a, 1, f_a) \in \mathcal{P}_3^{\frac{1}{5}, \frac{4}{5}}, \forall a \in \aset \nonumber \\ 
    & \text{previous linear constraints from \cref{FirstStage:aggregate1} to \cref{FirstStage:flow}} \label{eqn:ECRprev} \nonumber
\end{align}
\end{subequations}\vspace{0.5ex}
	}}
\end{center}

\subsection{Relationship between the dual formulation and maximum likelihood}\label{appendix:dual}

\renewcommand{\theequation}{B.\arabic{equation}}
\setcounter{equation}{0}

For illustration purposes, we focus on the entropy term related to the destination choice. 
Write the Lagrangian function including the parts related to \optjprobj, 
\begin{align*}
    & \lag (\jprobj, \scalej, \tripattrpar)  = - \frac{1}{\scalej} [\sum_{i \in \iset, j \in \jset} \optjprobj \ln (\frac{\optjprobj}{\obsprobi}) ] 
     + \sum_{k \in \kset} \tripattrpar [\sum_{i \in \iset, j \in \jset} \optjprobj \tripattr - \obsjprobj \tripattr] .
\end{align*}

Using \cref{eqn:c_prob_destin} and \cref{eqn:utildest}, replace $\ln \optcprobj$ with the following equation. 
\begin{align*} \textstyle
    \ln \optcprobj = \scalej (\sum_{k \in \kset} \tripattrpar \tripattr +S_{ij} ) - \ln {\sum_{j' \in \jset} \exp\left(\scalej (\sum_{k \in \kset} \tripattrpar \widehat{X}_{ij'}^k + S_{ij'})\right)}, 
\end{align*}

\begin{align*}
    & \lag(\jprobj, \scalej, \tripattrpar) \\ 
    & = - \sum_{i \in \iset, j \in \jset} \jprobj (S_{ij} + \sum_{k \in \kset} \tripattrpar \tripattr) + \frac{1}{\scalej} [\sum_{i \in \iset, j \in \jset} \jprobj \ln {\sum_{j' \in \jset} \exp\left(\scalej (S_{ij'} + \sum_{k \in \kset} \tripattrpar \widehat{X}^k_{ij'})\right)} ] \\
    & = - \frac{1}{\scalej} [\sum_{i \in \iset, j \in \jset} \jprobj \ln \frac{\exp\left(\scalej (S_{ij} + \sum_{k \in \kset} \tripattrpar \tripattr)\right)}{\sum_{j' \in \jset} \exp\left(\scalej (S_{ij'} + \sum_{k \in \kset} \tripattrpar \widehat{X}^k_{ij'})\right)} ] 
\end{align*}

Note the functional form of $p_{j|i}$ as defined in \cref{eqn:c_prob_destin} and \cref{eqn:utildest}.

\begin{align*}
    \max_{\jprobj, \tripattrpar, \scalej} \lag (\jprobj, \tripattrpar, \scalej) &= \max \left[ -\frac{1}{\scalej} \sum_{i \in \iset, j \in \jset} \jprobj \ln p_{j|i} (\scalej, \tripattrpar) \right] \\
    &= \min \left[ \frac{1}{\scalej} \sum_{i \in \iset, j \in \jset} \jprobj \ln p_{j|i} (\scalej, \tripattrpar) \right] 
\end{align*}

Write the dual problem.
\begin{align*}
    \min_{\jprobj, \tripattrpar, \scalej} \lag (\jprobj, \tripattrpar, \scalej) = \max \left[ \frac{1}{\scalej} \sum_{i \in \iset, j \in \jset} \jprobj \ln p_{j|i} (\scalej, \tripattrpar) \right] 
\end{align*}

This derivation illustrates the connection between the destination-choice component of the dual formulation and the log-likelihood objective.

\YS{\section{Background: Sequential Modeling Approach and Equivalent Convex Models}\label{section:preliminary}}

The concept of entropy is central to the combined convex modeling approach.
Historically, entropy was originally introduced in thermodynamics to quantify the amount of disorder in a system and was later adapted by \citet{shannon1948mathematical} for use in information theory.
The methods used in trip distribution and modal split generalize the principle of entropy maximization.

This section is structured as follows. First, we begin by exploring the entropy maximization model that corresponds to the gravity model in the trip distribution phase. Next, we analyze the entropy maximization model equivalent to the \gls{a:mnl} model in the modal split phase. Finally, we present the widely recognized convex optimization formulation for the Wardrop user equilibrium in the traffic assignment phase. These investigations demonstrate that convex optimization problems can effectively replace traditional fixed-point approaches in the trip distribution, modal split, and traffic assignment stages of transportation modeling.

\subsection{Convex formulation for trip distribution}

\citet{wilson1969use} proposed an entropy maximization model as an alternative to the conventional gravity model in the context of the trip distribution step. 
The entropy maximization approach aims to determine the most probable distribution of trips, and it turns out to serve as a theoretical foundation for the gravity model. 
Subsequently, the author proposes a multi-modal extension of the gravity model, where the modal split is implicitly incorporated into the generalized formulation (see \citet{wilson1969use} \S4 for details).

\subsubsection{Gravity model for trip distribution}
We first focus on the gravity model for a single mode of transport. 
We extend it to several transport modes later. 
The region is divided into zones: origin zone~$i \in \iset$ and destination zone~$j \in \jset$. 
The gravity model is used to estimate~$T_{ij}$, the number of trips between~$i$ and~$j$. Here, we deal with aggregated decision makers at spatial zone levels and do not distinguish individual travelers. 
The notations in this section follow the standard conventions used for the Gravity model in \citet{wilson1969use}.

In particular, we consider the gravity model that is expressed as
\begin{equation}
\label{eqn:gravity}
T_{ij} = A_i B_j O_i D_j f(c_{ij}), \forall i,j, 
\end{equation}
where
\begin{equation*} \textstyle
A_i = \frac{1}{\sum_{j} B_j D_j f(c_{ij})} \text{ and } B_j = \frac{1}{\sum_{j} A_i O_i f(c_{ij})},
\end{equation*}

where~$O_i$ is the total number of trip origins at~$i$ (i.e., trip production),~$D_j$ is the total number of trip destinations at~$j$ (i.e., trip attraction),~$c_{ij}$ is the generalized cost (e.g., a linear sum of fares and travel times),~$f(\cdot)$ is the impedance function, and~$A_i$ and~$B_j$ are the balancing factors for trip production and attraction.

\subsubsection{Entropy maximization principle in trip distribution}\label{sec:mostprobable}
Now, we suggest the entropy maximization model to get the most plausible distribution of trips. 
We eventually show that the gravity model can be obtained by solving the entropy maximization model. 
The basic assumption is that the probability of observing trip distribution $[T_{ij}]_{i \in \iset, j \in \jset}$ is proportional to the number of states of the system with such distribution. Based on combinatorial theory, we express the number of distinct arrangements as
\begin{equation*}
    \frac{\hat{T}!}{\prod_{ij} T_{ij}!},
\end{equation*}
where $\hat{T}$ is the total number of trips.

We formulate a problem to find the most probable trip distribution when trip productions ($[\hat{O}_i]_{i \in \iset}$), attractions ($[\hat{D}_j]_{j \in \jset}$), and the budget constraint ($\hat{C}$) are given. 

\begin{align*}
\begin{array}{c@{\quad}l}
\max_{T \in \mathbb R_+^{|\iset| \times |\jset|}} & \log \frac{\hat{T}!}{\prod_{ij} T_{ij}!}\\[2ex]
\text{s.t.} & \displaystyle \sum_{j \in \jset} T_{ij} = \hat{O}_i, \quad \forall i  \\[3ex]
            & \displaystyle \sum_{i \in \iset} T_{ij} = \hat{D}_j, \quad \forall j  \\[3ex]
            & \displaystyle \sum_{i \in \iset} \sum_{j \in \jset} T_{ij} \hat{c}_{ij} = \hat{C}  
\end{array}
\end{align*}

The first two constraints ensure that the total number of trips from origin~$i$ and to destination~$j$ should be equal to the given observations.
The third constraint ensures to satisfy the budget constraint where~$\hat{C}$ is the total amount spent on these trips in the region.

Using Stirling's approximation (i.e., $\log \bm{T}! \approx \bm{T} \log \bm{T} - \bm{T}$), the objective function can be re-written as 
\begin{equation*}
    \log \frac{\hat{T}!}{\prod_{i \in \iset, j \in \jset} T_{ij}!} \quad = \log \hat{T}! - \sum_{i \in \iset, j \in \jset} \log T_{ij}! 
    \quad \approx \log \hat{T}! - \sum_{i \in \iset}\sum_{j \in \jset} T_{ij} \log T_{ij} + \sum_{i \in \iset}\sum_{j \in \jset} T_{ij}.
\end{equation*}

Hence, we solve the approximated optimization problem as follows.
\begin{align*}
\begin{array}{c@{\quad}l}
\max_{T \in \mathbb R_+^{|\iset| \times |\jset|}} & - \sum_{i \in \iset}\sum_{j \in \jset} T_{ij} \log T_{ij} + \sum_{i \in \iset}\sum_{j \in \jset} T_{ij} \\[2ex]
\text{s.t.} & \displaystyle \sum_{j \in \jset} T_{ij} = \hat{O}_i, \quad \forall i  \\[3ex]
            & \displaystyle \sum_{i \in \iset} T_{ij} = \hat{D}_j, \quad \forall j  \\[3ex]
            & \displaystyle \sum_{i \in \iset} \sum_{j \in \jset} T_{ij} \hat{c}_{ij} = \hat{C}  
\end{array}
\end{align*}

Note that we call this approach an entropy maximization method due to the term $- \sum_{i \in \iset}\sum_{j \in \jset} T_{ij} \log T_{ij}$. 
Using $p_{ij} = T_{ij}/T$, we obtain $H = - \sum_{i \in \iset} \sum_{j \in \jset} p_{ij} \log p_{ij}$, which is the information theorist's definition of entropy.

The Lagrangian function $\mathcal{L}$ is given by
\begin{align}
    {\mathcal{L}} =& - \sum_{i \in \iset}\sum_{j \in \jset} T_{ij} \log T_{ij} + \sum_{i \in \iset}\sum_{j \in \jset} T_{ij} \nonumber \\ & + \sum_{i \in \iset}{\lambda_i} (\hat{O}_i - \sum_{j \in \jset} T_{ij}) 
    + \sum_{j \in \jset}{\mu_j} (\hat{D}_j - \sum_{i \in \iset} T_{ij}) 
    + \beta (\hat{C} - \sum_{i \in \iset} \sum_{j \in \jset} T_{ij} \hat{c}_{ij}),
\end{align}
where $\lambda_i, \mu_j, \beta$ are Lagrangian multipliers.

By solving the first order condition, $\frac{\partial {\mathcal{L}}}{\partial T_{ij}} = 0$, for all~$i, j$, we recover the gravity model where $f(c_{ij})= \exp(-\beta c_{ij})$. For a more comprehensive explanation, please refer to \citet{wilson1969use}. In conclusion, the distribution of trips with the highest probability closely aligns with the conventional gravity model distribution presented in \cref{eqn:gravity}. This statistical derivation constitutes a new theoretical base for the gravity model. Note that $\hat{C}$ does not need to be known as $\beta$ is given in practice. Without observation constraints, the trip distribution $[T_{ij}]_{i \in \iset, j \in \jset}$ would be trivial, with an equal number of trips for each entry. The constraints introduce non-uniformity to the trip distribution, ensuring consistency with observations.

\citet{wilson1969use} expanded the entropy maximization principle in the gravity model to encompass multiple modes of transportation. This revealed that the expanded model inherently incorporates a logit formulation of modal split. The subsequent section delves into the relationship between logit models and entropy maximizing models.

\subsection{Convex formulation for modal split}

We explore entropy maximizing models that are equivalent to logit choice models. The family of logit models is the standard way to model traveler mode choices. As a representative example, we show that \gls{a:mnl} can be derived from the entropy maximization model. 

\subsubsection{Multinomial logit model and maximum likelihood estimation for parameter estimation }\label{section:mnlmle}

\citet{mcfadden1973conditional} was the first to formulate travel mode and location decisions as problems in micro-economic consumer choices. In the random utility theory, the utility of an alternative $U^h_{m}$ can be expressed as the sum of the systematic component $V^h_m$ and the error component $\varepsilon^h_{m}$ associated with joint random variation across both individuals $h$ and alternatives $m$. 
\begin{align*}
    U^h_{m} = V^h_{m} + E^h_{m} = ASC_{m} + \sum_{k \in K} \beta_{k} \widehat{X}^h_{mk} + E^h_m,
\end{align*}
where $\widehat{X}^h_{mk}$ is the $k$-th trip-related attribute (e.g. travel cost or travel time) for alternative $m$, and $\beta_{k}$ is a parameter for $k$-th attribute. $ASC_{m}$ is the alternative specific constant for alternative $m$. $E^h_{m}$ is a realization of the error component $\varepsilon^h_{m}$. We focus on specific $ij$-pair and omit the subscript.

If $\varepsilon^h_m$ follows the type I extreme value distribution (also known as the Gumbel distribution) which is independent and identically distributed, the cumulative distribution function of an error term is
\begin{align}
    F([\varepsilon^h_m]_{m \in \mset}) = \mathbb{P} (E^h_m \leq \varepsilon^h_m, \forall m \in \mset) = \exp({-\exp({-{\theta \varepsilon^h_m}})}), \theta >0. \label{eqn:gumbel_mnl}
\end{align}

With the assumption of the Gumbel distribution, \citet{ben1985discrete} suggested that the choice probability can be expressed as a closed-form logit function through
\begin{align}\label{eqn:mnl_logit}
    P^h_{m}(\betavec; \theta) 
    = \mathbb{P}(U^h_{m} > U^{h}_{m'}, \forall m' \neq m) 
     = \frac{\exp({\theta V^h_{m}})}{\sum_{m'} \exp({\theta V^{h}_{m'}})} 
    = \frac{\exp({\theta (ASC_{m} + \sum_{k} \beta_{k} \widehat{X}^h_{mk})})}{\sum_{m'} \exp({\theta (ASC_{m'} \sum_{k} \beta_{k} \widehat{X}^h_{m'k})})}, \quad \forall m,h.
\end{align}
With the given observation of $y^h_m$, the indicator of whether individual $h$ chooses alternative $m$, we maximize the log-likelihood function to estimate the coefficient $\betavec$. The estimated coefficient $\hat{\betavec}$ is obtained as follows. 
\begin{align}\label{eqn:MLE}
    & \hat{\betavec} 
    = \argmax_{\betavec} \log \mathcal{L} = \argmax_{\betavec} \log \sum_h \sum_m y^h_m \log P^h_m(\betavec)
\end{align}
With the estimated $\hat{\betavec}$, we obtain ${\hat{V}^h_{m}}$, and then calculate $p^h_{m}$ through
\begin{align*}
    p^h_{m} &
    = \frac{\exp({\theta {\hat{V}^h_{m}}})}{\sum_{m'} \exp({\theta {\hat{V}^{h}_{m'}}})}, \quad \quad \forall m,h.
\end{align*}

\subsubsection{Entropy maximization principle in parameter estimations}\label{section:utility_maximize_MNL}

\citet{anas1983discrete} showed that the behavior demand model by \citet{mcfadden1973conditional} and the entropy maximizing model by \citet{wilson1969use}, should be seen as two equivalent views of the same problem. 
Consider the following entropy maximizing problem. 
\begin{subequations}
\begin{align}    \max_{[p^h_{m}]_{\forall h \in H, m \in M}} &  - \frac{1}{\theta}\sum_{h} \sum_{m} p^h_{m} \log p^h_{m} \label{maxentropy:objective}\\
    \text{s.t.} & \sum_{m} p^h_{m} = 1; \quad \forall h \quad [\lambda_h] \label{maxentropy:prob_sum}\\
    & \sum_{h} p^h_{m} = \sum_{h} \hat{y}^h_{m}; \quad \forall m \quad [\gamma_m] \label{maxentropy:gamma_m}\\
    & \sum_{m} \sum_{h} p^h_{m} \widehat{X}^h_{mk} = \sum_{m} \sum_{h} \hat{y}^h_{m} \widehat{X}^h_{mk}, \quad \forall k \quad [\alpha_k] \label{maxentropy:alpha_0}.
\end{align}
\end{subequations}

The objective function (\ref{maxentropy:objective}) seeks the entropy maximizing predictions of $p^h_m$. 
Constraint (\ref{maxentropy:prob_sum}) requires the choice probability of choosing each alternative should be summed up to one for each individual. Constraint (\ref{maxentropy:gamma_m}) ensures that the predicted number of choosing each alternative equals the actual number of choosing it. Constraint (\ref{maxentropy:alpha_0}) ensures that the expectation of the aggregated deterministic value of each alternative should equal the observed aggregated value. Constraint (\ref{maxentropy:gamma_m}) and (\ref{maxentropy:alpha_0}) both ensure that the predictions replicate the aggregated observations on the entire system.

By forming the Lagrangian function of the above formulation and by solving the first order condition, we obtain the form of \gls{a:mnl} as
\begin{align*}
    p^h_{m}  
    = \frac{\exp({\theta (\gamma_m + \sum_{k}\alpha_k \widehat{X}^{h}_{mk})})}{\sum_{m'} \exp({\theta (\gamma_{m'} + \sum_{k}\alpha_k \widehat{X}^{h}_{m'k})})}.
\end{align*}

Moreover, \citet{anas1983discrete} proved that the \gls{a:mnl} model can be identically estimated via maximum likelihood in \cref{eqn:MLE} that finds $\betavec$ and $\textbf{ASC}$ in \cref{eqn:mnl_logit}, or via solving the entropy maximization problem to find Lagrangian multipliers $\alphavec$ and $\gammavec$, thus $\betavec = \alphavec$ and $\textbf{ASC} = \gammavec$. The maximum entropy estimators reproduce rational user behavior.

A special case deals with aggregated travel patterns when the parameters are predetermined. \citet{miyagi1996direct} proposed a satisfaction maximization problem. The satisfaction function, denoted as $\satis(\utilvec)$, is the sum of the deterministic utility and the entropy function, and it is formulated as
\begin{align*}
    \min_{\bm{p}} \satis(\utilvec) 
    = \min_{\bm{p}} \sum_{m \in \mathcal{M}} \hat{V}_m p_m - \frac{1}{\theta} \sum_{m \in \mathcal{M}} p_m \ln (p_m).
\end{align*}
This formulation represents a special case of the entropy maximization problem discussed earlier, where the parameters are specified as $\hat{V}_m = \hat{ASC}_m + \sum_{k} \hat{\beta}_{k} \hat{X}_{mk}$, and the individual variable $h$ is omitted to reflect the aggregated nature of the travel patterns.

\subsection{Convex formulation for traffic assignment}

% \subsubsection{Beckmann equation and User Equilibrium}

The Beckmann equation is a fundamental tool in static traffic assignment, offering an optimization framework to describe traffic flow at equilibrium \citep{beckmann1956studies}. Formally, the Beckmann model is expressed as an optimization problem, where the objective function represents the integral of total travel time across all links in the network as
\begin{align}
    \min_{\bm{f}} \sum_{a \in A} \int_0^{f_a} t_a(x) dx,
\end{align}
where \(f_a\) is the flow count on link \(a\), and \(t_a(x)\) is the latency function for link \(a\), which depends on the flow count. Under the assumption that the latency function is increasing (i.e., \( \frac{dt_a}{df_a} > 0 \)), the objective function is convex. This convexity ensures that solving the optimization problem characterizes the user equilibrium (UE) condition. In user equilibrium, no individual user has an incentive to deviate from their current route choice, as each user selects a route that minimizes their own travel time. Consequently, the system is in equilibrium when every user experiences the shortest travel time possible given the route choices of others.

Under this rationale, two prominent approaches exist for solving the Beckmann equation to achieve user equilibrium. The first approach involves simulation such as EMME4. Simulation-based analysis offers flexibility in capturing more realistic traffic patterns. It is relatively straightforward to utilize existing software, making it well-suited for large-scale network problems, where it can model complex interactions in transportation systems effectively. 

The second approach is more analytical, involving the reformulation of the Beckmann problem as a second-order cone programming (SOCP) problem. Although often overlooked, the Beckmann equation can indeed be reformulated as an \gls{a:socp}, provided that the road latency function can be expressed as a convex power function \citep{wei2019efficient}. This observation is crucial for scalability, as \gls{a:socp} formulations are nearly as tractable as \gls{a:lp} due to recent advancements in interior-point methods.

\subsection{Related concepts in other disciplines}

% \SSC{This section is a great idea, but currently just a few bullet points. This should be expanded to include a bit more detail and also connect it to what this paper is doing.}

Similar concepts have, in fact, already been explored in fields beyond transportation, such as information theory and operations research. To put our work in the context of these efforts, facilitate interdisciplinary research exchange and offer additional insights, we present some relevant terminology from other domains along with corresponding references.
The well-known relationship between random utility models and entropy functions (within the transportation community) has been investigated through the lens of convex conjugate duality in operations research \citep{muller2022discrete}, opening the door to more generalized models. 
\citet{fosgerau2022inverse} introduced a generalization of Shannon's entropy, capturing the flexible substitution patterns as well as complementarity.
Furthermore, the classic relationship between the \gls{a:mnl} and maximum likelihood estimation (see \cref{section:mnlmle}) can be viewed as a special case of the relationship between the \gls{a:kl} divergence and the maximum likelihood estimation \citep{polyanskiy2024information}. 
Minimizing the \gls{a:kl} divergence leads to minimizing the statistical distance between the observed and target distributions, which is equivalent to applying maximum likelihood estimation.

\end{appendices}

\end{document}